\documentclass[a4paper,UKenglish,cleveref,hyperref,autoref]{lipics-v2021}
\title{Twin-width VI: the lens of contraction sequences}
\titlerunning{Twin-width VI: the lens of contraction sequences}%optional, please use if title is longer than one line

\author{\'{E}douard Bonnet}{Univ Lyon, CNRS, ENS de Lyon, Université Claude Bernard Lyon 1, LIP UMR5668, France \and \url{http://perso.ens-lyon.fr/edouard.bonnet/}}{edouard.bonnet@ens-lyon.fr}{https://orcid.org/0000-0002-1653-5822}{}
\author{Eun Jung Kim}{Universit\'{e} Paris-Dauphine, PSL University, CNRS UMR7243, LAMSADE, Paris, France}{eun-jung.kim@dauphine.fr}{https://orcid.org/0000-0002-6824-0516}{}
\author{Amadeus Reinald}{Univ Lyon, CNRS, ENS de Lyon, Université Claude Bernard Lyon 1, LIP UMR5668, France}{amadeus.reinald@ens-lyon.fr}{https://orcid.org/0000-0002-8108-4036}{}
\author{St\'{e}phan Thomass\'{e}}{Univ Lyon, CNRS, ENS de Lyon, Universit\'{e} Claude Bernard Lyon 1, LIP UMR5668, France}{stephan.thomasse@ens-lyon.fr}{}{}

\authorrunning{\'E. Bonnet, E. J. Kim, A. Reinald, S. Thomassé}
\Copyright{Édouard Bonnet, Eun Jung Kim, Amadeus Reinald, Stéphan Thomassé}%TODO mandatory, please use full first names. LIPIcs license is "CC-BY";  http://creativecommons.org/licenses/by/3.0/

%TODO mandatory: Please choose ACM 2012 classifications from https://dl.acm.org/ccs/ccs_flat.cfm 

%TODO mandatory; please add comma-separated list of keywords

\category{}%optional, e.g. invited paper

\relatedversion{}%optional, e.g. full version hosted on arXiv, HAL, or other respository/website
\supplement{}%optional, e.g. related research data, source code, ... hosted on a repository like zenodo, figshare, GitHub, ...

\funding{
  This paper was supported by the ANR projects (French National Research Agency) TWIN-WIDTH (ANR-21-CE48-0014-01) and Digraphs (ANR-19-CE48-0013-01).
  E. J. K. was also supported by the ANR project ASSK (ANR-18-CE40-0025-01).}
\acknowledgements{}%optional

\nolinenumbers %uncomment to disable line numbering

\hideLIPIcs  %uncomment to remove references to LIPIcs series (logo, DOI, ...), e.g. when preparing a pre-final version to be uploaded to arXiv or another public repository

\EventEditors{John Q. Open and Joan R. Access}
\EventNoEds{2}
\EventLongTitle{42nd Conference on Very Important Topics (CVIT 2016)}
\EventShortTitle{CVIT 2016}
\EventAcronym{CVIT}
\EventYear{2016}
\EventDate{December 24--27, 2016}
\EventLocation{Little Whinging, United Kingdom}
\EventLogo{}
\SeriesVolume{42}
\ArticleNo{23}

\usepackage[utf8]{inputenc}  

\usepackage[T1]{fontenc}
\usepackage{lmodern}

\usepackage[colorinlistoftodos,bordercolor=orange,backgroundcolor=orange!20,linecolor=orange,textsize=normalsize]{todonotes}

\usepackage{amsmath}  
\usepackage{amssymb}     
\usepackage{bbm}
\usepackage{accents}

\usepackage{complexity}

\usepackage{hyperref}
\usepackage{mathrsfs}
\usepackage{paralist}
\usepackage{fixmath}

\usepackage{pifont}% http://ctan.org/pkg/pifont
\newcommand{\cmark}{\ding{51}}%
\newcommand{\xmark}{\ding{55}}%

\crefformat{equation}{\textup{#2(#1)#3}}
\crefrangeformat{equation}{\textup{#3(#1)#4--#5(#2)#6}}
\crefmultiformat{equation}{\textup{#2(#1)#3}}{ and \textup{#2(#1)#3}}
{, \textup{#2(#1)#3}}{, and \textup{#2(#1)#3}}
\crefrangemultiformat{equation}{\textup{#3(#1)#4--#5(#2)#6}}%
{ and \textup{#3(#1)#4--#5(#2)#6}}{, \textup{#3(#1)#4--#5(#2)#6}}{, and \textup{#3(#1)#4--#5(#2)#6}}

\Crefformat{equation}{#2Equation~\textup{(#1)}#3}
\Crefrangeformat{equation}{Equations~\textup{#3(#1)#4--#5(#2)#6}}
\Crefmultiformat{equation}{Equations~\textup{#2(#1)#3}}{ and \textup{#2(#1)#3}}
{, \textup{#2(#1)#3}}{, and \textup{#2(#1)#3}}
\Crefrangemultiformat{equation}{Equations~\textup{#3(#1)#4--#5(#2)#6}}%
{ and \textup{#3(#1)#4--#5(#2)#6}}{, \textup{#3(#1)#4--#5(#2)#6}}{, and \textup{#3(#1)#4--#5(#2)#6}}

\usepackage{xspace}
\usepackage{tikz}

\usepackage[ruled,vlined,linesnumbered]{algorithm2e}

\usetikzlibrary{fit}
\usetikzlibrary{arrows}
\usetikzlibrary{patterns}
\usetikzlibrary{calc}
\usetikzlibrary{shapes}
\usetikzlibrary{positioning}
\usetikzlibrary{math}
\usetikzlibrary{shapes.symbols}

\usepackage[scr=boondox,scrscaled=1.05]{mathalfa}

\renewcommand{\geq}{\geqslant}
\renewcommand{\leq}{\leqslant}
\renewcommand{\preceq}{\preccurlyeq}

\newcommand{\equi}{functionally equivalent\xspace}

\newcommand{\adj}[2]{\text{Adj}_{#1}(#2)}

\DeclareMathOperator*{\bigland}{\bigwedge}

\newtheorem*{theorem*}{Theorem}
\theoremstyle{definition}
\newenvironment{proofofclaim}{\noindent \textsc{Proof of the Claim:}}{\hfill$\Diamond$\medskip}

\newcommand{\tww}{\text{tww}}
\newcommand{\ttww}{\text{ttww}}
\newcommand{\ltww}{\text{ltww}}
\newcommand{\otww}{\text{otww}}

\renewcommand{\P}{{\mathcal P}}
\newcommand{\MC}{{\mathcal C}}

\newcommand{\cC}{\mathcal{C}}

\colorlet{npink}{red!30!pink}

\newcommand{\local}{component\xspace}
\newcommand{\total}{total\xspace}
\newcommand{\Local}{Component\xspace}
\newcommand{\Total}{Total\xspace}

\newcommand{\cbd}[2]{\mathcal D_{#1,#2}}
\newcommand{\cbe}[2]{\mathcal E_{#1,#2}}

\newcommand{\eqfo}[1]{\equiv^{\text{FO}}_{#1}}
\newcommand{\eqmso}[1]{\equiv^{\text{MSO}}_{#1}}

\newcommand{\ltp}[1]{\text{loc-mso-tp}_{#1}}
\newcommand{\tp}[1]{\text{mso-tp}_{#1}}

\newcommand\abs[1]{\lvert #1\rvert}

\newenvironment{edouard}[1]{\textcolor{blue}{Édouard:~#1}}{}
\newenvironment{ejk}[1]{\textcolor{orange}{EJK:~#1}}{}

\begin{document}

\maketitle

\begin{abstract}
  A contraction sequence of a graph consists of iteratively merging two of its vertices until only one vertex remains.
  The recently introduced twin-width graph invariant is based on contraction sequences.
  More precisely, if one puts error edges, henceforth red edges, between two vertices representing non-homogeneous subsets, the twin-width is the minimum integer~$d$ such that a contraction sequence exists that keeps red degree at most~$d$.
  By changing the condition imposed on the trigraphs (i.e., graphs with some edges being red) and possibly slightly tweaking the notion of contractions, we show how to characterize the well-established bounded rank-width, tree-width, linear rank-width, path-width --usually defined in the framework of branch-decompositions--, and proper minor-closed classes by means of contraction sequences.
  
 Contraction sequences hold a crucial advantage over branch-decompositions: While one can scale down contraction sequences to capture classical width notions, the more general bounded twin-width goes beyond their scope, as it contains planar graphs in particular, a class with unbounded rank-width.
As an application we give a transparent alternative proof of the celebrated Courcelle's theorem (actually of its generalization by Courcelle, Makowsky, and Rotics), that MSO$_2$ (resp. MSO$_1$) model checking on graphs with bounded tree-width (resp. bounded rank-width) is fixed-parameter tractable in the size of the input sentence.
  We are hopeful that our characterizations can help in other contexts.

We then explore new avenues along the general theme of contraction sequences both in order to refine the landscape between bounded tree-width and bounded twin-width (via spanning twin-width) and to capture more general classes than bounded twin-width. To this end, we define an oriented version of twin-width, where appearing red edges are oriented away from the newly contracted vertex, and the mere red out-degree should remain bounded.
  Surprisingly, classes of bounded oriented twin-width coincide with those of bounded twin-width.
  This greatly simplifies the task of showing that a class has bounded twin-width.
  As an example, using a lemma by Norine, Seymour, Thomas, and Wollan, we give a 5-line proof that $K_t$-minor free graphs have bounded twin-width.
  Without oriented twin-width, this fact was shown by a somewhat intricate 4-page proof in the first paper of the series.
  Finally we explore the concept of partial contraction sequences, where, instead of terminating on a single-vertex graph, the sequence ends when reaching a particular target class.
  We show that FO model checking (resp.~$\exists$FO model checking) is fixed-parameter tractable on classes with partial contraction sequences to a class of bounded degree (resp.~bounded expansion), provided such a sequence is given.
  Efficiently finding such partial sequences could turn out simpler than finding a (complete) sequence. 
\end{abstract}

\maketitle

\newpage

\section{Introduction}\label{sec:intro}

A~\emph{trigraph} is a graph with some of its edges being distinguished, typically called \emph{red edges}, while the rest of the edges are called \emph{black}.
The (vertex) \emph{contraction} (or \emph{identification}) of two non-necessarily adjacent vertices $u$ and $v$ in a trigraph consists of merging these two vertices into a new vertex $w$, keeping every edge $wx$ black if both $ux$ and $vx$ were black edges, and turning all the other edges incident to $w$ red.
The rest of the trigraph does not change.
A~\emph{contraction sequence} of an $n$-vertex (tri)graph $G$ is a sequence of trigraphs $G=G_n, G_{n-1}, \ldots, G_1$ such that $G_i$ is an $i$-vertex trigraph, obtained by performing one contraction in $G_{i+1}$.
A~\emph{$d$-sequence} is a contraction sequence such that every trigraph of the sequence has maximum red degree at most~$d$.
The \emph{twin-width} of a graph is defined via contraction sequences: It is the minimum integer $d$ such that $G$ admits a $d$-sequence.
See~\cref{fig:contraction-sequence} for an example of a graph with a 2-sequence.
\begin{figure}[h!]
  \centering
  \resizebox{400pt}{!}{
  \begin{tikzpicture}[
      vertex/.style={circle, draw, minimum size=0.68cm}
    ]
    \def\s{1.2}
    %G=G7
    \foreach \i/\j/\l in {0/0/a,0/1/b,0/2/c,1/0/d,1/1/e,1/2/f,2/1/g}{
      \node[vertex] (\l) at (\i * \s,\j * \s) {$\l$} ;
    }
    \foreach \i/\j in {a/b,a/d,a/f,b/c,b/d,b/e,b/f,c/e,c/f,d/e,d/g,e/g,f/g}{
      \draw (\i) -- (\j) ;
    }

    %G6=G7/e,f
    \begin{scope}[xshift=3 * \s cm]
    \foreach \i/\j/\l in {0/0/a,0/1/b,0/2/c,1/0/d,2/1/g}{
      \node[vertex] (\l) at (\i * \s,\j * \s) {$\l$} ;
    }
    \foreach \i/\j/\l in {1/1/e,1/2/f}{
      \node[vertex,opacity=0.2] (\l) at (\i * \s,\j * \s) {$\l$} ;
    }
    \node[draw,rounded corners,inner sep=0.01cm,fit=(e) (f)] (ef) {ef} ;
    \foreach \i/\j in {a/b,a/d,b/c,b/d,b/ef,c/ef,c/ef,d/g,ef/g,ef/g}{
      \draw (\i) -- (\j) ;
    }
    \foreach \i/\j in {a/ef,d/ef}{
      \draw[red, very thick] (\i) -- (\j) ;
    }
    \end{scope}

    %G5=G6/a,d
    \begin{scope}[xshift=6 * \s cm]
    \foreach \i/\j/\l in {0/1/b,0/2/c,2/1/g,1/1/ef}{
      \node[vertex] (\l) at (\i * \s,\j * \s) {$\l$} ;
    }
    \foreach \i/\j/\l in {0/0/a,1/0/d}{
      \node[vertex,opacity=0.2] (\l) at (\i * \s,\j * \s) {$\l$} ;
    }
    \draw[opacity=0.2] (a) -- (d) ;
    \node[draw,rounded corners,inner sep=0.01cm,fit=(a) (d)] (ad) {ad} ;
    \foreach \i/\j in {ad/b,b/c,b/ad,b/ef,c/ef,c/ef,ef/g,ef/g}{
      \draw (\i) -- (\j) ;
    }
    \foreach \i/\j in {ad/ef,ad/g}{
      \draw[red, very thick] (\i) -- (\j) ;
    }
    \end{scope}

     %G4=G5/b,ef
    \begin{scope}[xshift=9 * \s cm]
    \foreach \i/\j/\l in {0/2/c,2/1/g,0.5/0/ad}{
      \node[vertex] (\l) at (\i * \s,\j * \s) {$\l$} ;
    }
    \foreach \i/\j/\l in {0/1/b,1/1/ef}{
      \node[vertex,opacity=0.2] (\l) at (\i * \s,\j * \s) {$\l$} ;
    }
    \draw[opacity=0.2] (b) -- (ef) ;
    \node[draw,rounded corners,inner sep=0.01cm,fit=(b) (ef)] (bef) {bef} ;
    \foreach \i/\j in {ad/bef,bef/c,bef/ad,c/bef,c/bef,bef/g}{
      \draw (\i) -- (\j) ;
    }
    \foreach \i/\j in {ad/bef,ad/g,bef/g}{
      \draw[red, very thick] (\i) -- (\j) ;
    }
    \end{scope}

     %G3=G4/a,dg
    \begin{scope}[xshift=11.7 * \s cm]
    \foreach \i/\j/\l in {0/2/c}{
      \node[vertex] (\l) at (\i * \s,\j * \s) {$\l$} ;
    }
     \foreach \i/\j/\l in {0.5/0/adg,0.5/1.1/bef}{
      \node[vertex] (\l) at (\i * \s,\j * \s) {\footnotesize{\l}} ;
    }
    \foreach \i/\j in {c/bef}{
      \draw (\i) -- (\j) ;
    }
    \foreach \i/\j in {adg/bef}{
      \draw[red, very thick] (\i) -- (\j) ;
    }
    \end{scope}

    %G2=G3/c,bef
    \begin{scope}[xshift=13.7 * \s cm]
    \foreach \i/\j/\l in {0.5/0/adg,0.5/1.1/bcef}{
      \node[vertex] (\l) at (\i * \s,\j * \s) {\footnotesize{\l}} ;
    }
    \foreach \i/\j in {adg/bcef}{
      \draw[red, very thick] (\i) -- (\j) ;
    }
    \end{scope}

    %G1=G2/adg,bcef
    \begin{scope}[xshift=15 * \s cm]
    \foreach \i/\j/\l in {1/0.75/abcdefg}{
      \node[vertex] (\l) at (\i * \s,\j * \s) {\tiny{\l}} ;
    }
    \end{scope}
    
  \end{tikzpicture}
  }
  \caption{A 2-sequence witnessing that the initial graph has twin-width at most~2.}
  \label{fig:contraction-sequence}
\end{figure}
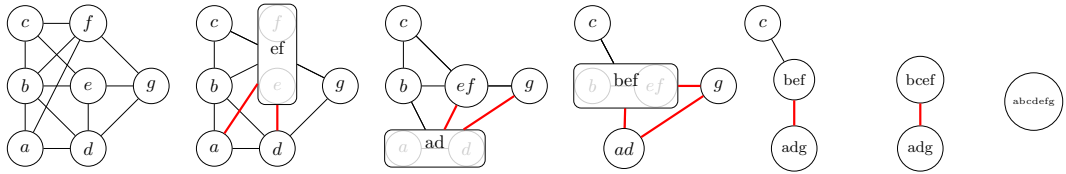
Not to hinder the flow of this introduction, we will try to limit further definitions.
The reader is deferred to~\cref{sec:prelim} if encountering some unknown terminology.\footnote{If a term is still not defined there, it is not important for the rest of the paper.}

Classes of bounded twin-width are surprisingly diverse.
They include for instance classes of bounded tree-width, or even bounded rank-width, proper minor-closed classes, hereditary proper subclasses of permutation graphs, subgraphs of $O(1)$-dimensional grids~\cite{twin-width1}, as well as $\Omega(\log n)$-subdivisions of $n$-vertex graphs, classes with bounded queue or stack number, and some families of expanders~\cite{twin-width2}.
Nevertheless classes of bounded twin-with have interesting properties: If a class $\mathcal C$ has bounded twin-width, then so has any class that is obtained from $\mathcal C$ by a fixed first-order transduction~\cite{twin-width1}, they are $\chi$-bounded~\cite{twin-width3}, and allow a fixed-parameter tractable ({\FPT}) algorithm for FO model checking~\cite{twin-width1}, provided $O(1)$-sequences are given in input, and more practical {\FPT} algorithms on specific problems like \textsc{$k$-Independent Set} or \textsc{$k$-Dominating Set}~\cite{twin-width3}.
Twin-width naturally extends to matrices over finite alphabets and binary structures in general~\cite{twin-width1,twin-width&permutations}.
Efficiently approximating twin-width (that is, returning an $f(d)$-sequence when the twin-width of the input is at most~$d$) can be done for classes of \emph{totally ordered} binary structures~\cite{twin-width4}, and for all the above-mentioned classes of bounded twin-width, but remains an open challenge for \emph{unordered} graphs.

\medskip

This paper investigates variations on the theme of contraction sequences.
We first show that bounded rank-width and bounded linear rank-width can be defined by means of contractions sequences.
Instead of requiring the maximum red degree to be bounded, one shall strengthen the condition to bounded-size red components, and bounded total number of red edges,\footnote{modulo the technicality of introducing red loops (see~\cref{sec:prelim})} respectively.
In the sparse regime of biclique-free classes, this characterizes bounded tree-width and bounded path-width.
See~\cref{subfig:ltww,subfig:ttww} for some illustration.

Let us elaborate on that with an alternative formalism.
A useful equivalent viewpoint on contraction sequences is the notion of \emph{partition sequence}, that is, a sequence $\mathcal P_n, \ldots, \mathcal P_1$ of partitions of the vertex set $V(G)$, where $\mathcal P_n$ is the partition into singletons, $\mathcal P_1=\{V(G)\}$, and each $\mathcal P_i$ is obtained by merging two parts of $\mathcal P_{i+1}$.
Then a \emph{width} $w: \mathcal P(V(G)) \to \mathbb N$ (i.e., a function from the partitions of $V(G)$ to the naturals) can be naturally lifted to the graph $G$ as the minimum integer $t$ such that there exists a partition sequence $\mathcal P_n, \ldots, \mathcal P_1$ satisfying $w(\mathcal P_i) \leqslant t$ for every $i$.
This defines the \emph{width of $G$ associated to} $w$.

What are partitions of ``good quality'' or \emph{small width}?
Probably partitions $\mathcal P_i$ such that the quotient $G/\mathcal P_i$ ``captures'' the edge set of $G$ quite well, minimizing the extra amount of information one needs to fully recover $G$.
In that respect, the ideal scenario is when a pair of parts $X,Y \in \mathcal P_i$ is \emph{homogeneous}, that is, $X$ and $Y$ are completely adjacent or completely non-adjacent in $G$.
Consider the auxiliary graph with vertex set the parts of $\mathcal P_i$, and edge set all pairs of parts $X,Y \in \mathcal P_i$, which are \emph{not} homogeneous.
Note that this auxiliary graph is precisely the \emph{red graph} of trigraph $G_i$ (as previously defined), that is, obtained from $G_i$ by keeping the red edges only.
%A natural convention is to add a loop $X,X$ to every part $X$ which is not a singleton (since equality is neither an edge nor a non edge).
%Loops count as degree 1.
Three width measures come relatively naturally: the maximum degree, $w_d$, the maximum component size, $w_c$, and the total number of edges, $w_t$, in the red graph of $G_i$.
We obtain the following invariants on graphs.
\begin{compactitem}
\item The \emph{twin-width} of $G$ as the width $\tww(G)$ associated to $w_d$ (see~\cref{subfig:tww}).
\item The \emph{\local twin-width} of $G$ as the width $\ltww(G)$ associated to $w_c$ (see~\cref{subfig:ltww}).
\item The \emph{\total twin-width} of $G$ as the width $\ttww(G)$ associated to $w_t$ (see~\cref{subfig:ttww}).
\end{compactitem}

\begin{figure}[h!]
  \centering
    \begin{subfigure}[t]{0.4\textwidth}
        \centering
        \begin{tikzpicture}[vertex/.style={draw,circle,inner sep=0.08cm},err/.style={red,very thick}]
          \def\t{6}
          \pgfmathtruncatemacro\tm{\t-1}
          \def\s{0.75}
          %vertices
          \foreach \i in {1,...,\t}{
            \foreach \j in {1,...,\t}{ 
              \node[vertex] (v\i\j) at (\i * \s,\j * \s) {} ;
            }
          }
          %red edges
          \foreach \i [count = \ip from 2] in {1,...,\tm}{
            \foreach \j in {1,...,\t}{
              \draw[err] (v\i\j) -- (v\ip\j) ;
              \draw[err] (v\j\i) -- (v\j\ip) ;
            }
          }
          %black edges
          \foreach \i/\j/\k/\l in {1/1/2/2,1/3/3/4,1/4/2/5, 2/1/1/2,2/1/3/2,2/5/1/6, 3/2/4/3,3/4/4/3, 4/1/5/2,4/2/5/1,4/4/5/5,4/5/5/4, 5/3/6/4,5/6/6/5}{
            \draw (v\i\j) -- (v\k\l) ;
          }
        \end{tikzpicture}
        \caption{Twin-width: red degree.}
        \label{subfig:tww}
    \end{subfigure}%
    ~ 
    \begin{subfigure}[t]{0.6\textwidth}
        \centering
         \begin{tikzpicture}[vertex/.style={draw,circle,inner sep=0.08cm},err/.style={red,very thick}]
          \def\t{6}
          \pgfmathtruncatemacro\tm{\t-1}
          \def\s{1.1}
          %vertices
          \foreach \i in {1,...,\t}{
            \foreach \j in {1,...,\t}{ 
              \node[vertex] (v\i\j) at (\i * \s,0,\j * \s) {} ;
            }
          }
          \node[vertex] (u) at (\t/2 * \s + \s/2,2.5,\t/2 * \s + \s/2) {} ;
          \node[vertex] (u2) at (\t/2 * \s + 3 * \s/2,2.5,\t/2 * \s + \s/2) {} ;
          %red arcs
          \foreach \i in {1,...,\t}{
            \foreach \j in {1,...,\t}{
              \pgfmathparse{random(5)}
              \ifnum\pgfmathresult=1
              \draw[err,->,opacity=0.5] (v\i\j) -- (u) ;
              \fi
              \pgfmathparse{random(4)}
              \ifnum\pgfmathresult=1
              \draw[err,->,opacity=0.5] (v\i\j) -- (u2) ;
              \fi
            }
          }  
          \foreach \i [count = \ip from 2] in {1,...,\tm}{
            \foreach \j in {1,...,\t}{
              \pgfmathparse{random(3)}
              \ifnum\pgfmathresult=1
              \draw[err,->] (v\i\j) -- (v\ip\j) ;
              \else
              \ifnum\pgfmathresult=2
              \draw[err,<-] (v\i\j) -- (v\ip\j) ;
              \else
              \draw[err,<->] (v\i\j) -- (v\ip\j) ;
              \fi
              \fi
              \pgfmathparse{random(3)}
              \ifnum\pgfmathresult=1
              \draw[err,->] (v\j\i) -- (v\j\ip) ;
              \else
              \ifnum\pgfmathresult=2
              \draw[err,<-] (v\j\i) -- (v\j\ip) ;
              \else
              \draw[err,<->] (v\j\i) -- (v\j\ip) ;
              \fi
              \fi
            }
          }
           %black edges
          \foreach \i/\j/\k/\l in {1/3/3/4,1/4/3/5, 2/1/1/3,2/1/3/2,2/5/1/6, 3/2/4/3,3/4/4/3, 4/1/5/2,4/2/5/1,4/4/5/5,4/5/5/4, 5/3/6/4,5/5/6/6,5/6/6/5}{
            \draw (v\i\j) -- (v\k\l) ;
          }
        \end{tikzpicture}
        \caption{Oriented twin-width: red out-degree.}
        \label{subfig:otww}
    \end{subfigure}

    \begin{subfigure}[t]{0.51\textwidth}
        \centering
        \begin{tikzpicture}[vertex/.style={draw,circle,inner sep=0.08cm},err/.style={red,very thick}]
          \def\s{0.75}
          %vertices
          \foreach \i/\j/\l in {0/0/a,0/1/b,1/0/c,1/1/d, 1/2/e,1/3/f,2/2/g,2/3/h, 2/0/i,2/1/j,3/0/k,3/1/l, 2.5/4/m,3/5/n,3.5/4/o, 4.5/3/p, 4/2/q,3.75/1/r,4.25/1/s, 5/2/t,4.75/1/u,5.25/1/v}{
            \node[vertex] (\l) at (\i * \s, \j * \s) {} ;
          }
          %red edges
          \foreach \i/\j in {a/b,a/c,a/d,b/c,b/d,c/d, i/j,i/l,i/k, r/s,r/q, u/v,u/t,v/t, f/h}{
            \draw[err] (\i) -- (\j) ;
          }
          %black edges
          \foreach \i/\j in {b/e,d/e, j/g,l/g, e/g,e/h,e/f,g/f,g/h, f/m,h/m, m/o,m/n,o/n, o/p, p/q,q/t,p/t, q/s}{
            \draw (\i) -- (\j) ;
          }
        \end{tikzpicture}
        \caption{\Local twin-width: size of red components.}
        \label{subfig:ltww}
    \end{subfigure}%
    ~ 
    \begin{subfigure}[t]{0.49\textwidth}
        \centering
        \begin{tikzpicture}[vertex/.style={draw,circle,inner sep=0.08cm},err/.style={red,very thick}, every loop/.style={}]
          \def\r{1.3}
          \def\rp{1.8}
          \foreach \a in {20,40,...,360}{
            \node[vertex] (p\a) at (\a:\r) {} ;
            \node[vertex] (q\a) at (\a:\rp) {} ;
          }
          %red edges
          \foreach \i/\j in {p260/q260,p260/q280,q260/q280,p280/q260,p300/q300}{
            \draw[err] (\i) -- (\j) ;
          }
          \foreach \i in {p280,p300}{
            \path[err] (\i) edge [loop above] (\i) ;
          }
          \foreach \i in {q260,q280}{
            \path[err] (\i) edge [loop below] (\i) ;
          }
          %black edges
          \foreach \i in {20,40,...,340}{
            \pgfmathtruncatemacro\ip{\i+20}
            \draw (p\i) -- (p\ip) ;
          }
           \foreach \i in {20,40,...,240,280,300,...,340}{
            \pgfmathtruncatemacro\ip{\i+20}
            \draw (q\i) -- (q\ip) ;
           }
           \foreach \i/\j in {p360/p20,q360/q20,p360/q20, p40/q40,p40/q80,q60/p80, p100/q120,p100/q100,p120/q120,p120/q100, p140/q160,p160/q180,p180/q200,p200/q200, p220/q240,p240/q220, q280/p300, p320/q340,p320/q320,p340/q320}{
            \draw (\i) -- (\j) ;
          }
        \end{tikzpicture}
        \caption{\Total twin-width: number of red edges.}
        \label{subfig:ttww}
    \end{subfigure}
    \caption{Examples of graphs with low value of the four corresponding twin-width variants.}
    \label{fig:all-tww}
\end{figure}
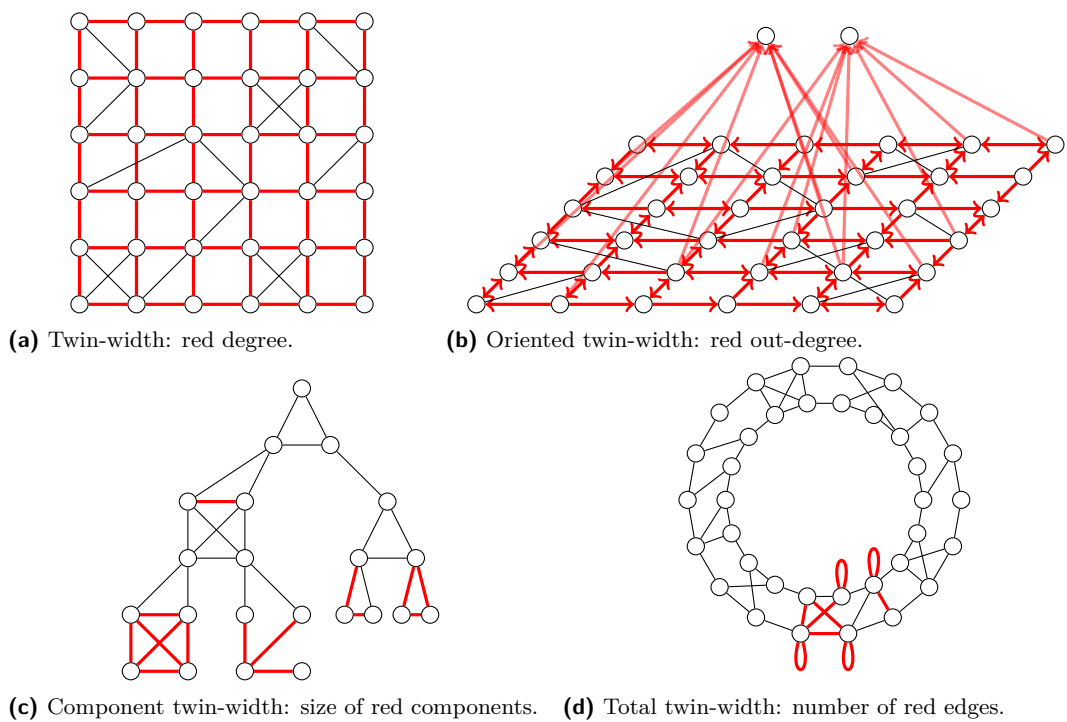

%Note that there is a tiny variation with our original definition of twin-width since loops add one to degree.
%In particular cographs have now twin-width 1.
%Observe also that $tww(G)\leq ltww(G)\leq ttww(G)$.

As we will see in~\cref{sec:classical}, the striking fact is that the latter two parameters already exist, up to functional equivalence. 

\begin{theorem}\label{thm:cs-width}
  The following parameters are \equi:
\begin{compactitem}
\item \Local twin-width and rank-width;
\item \Total twin-width and linear rank-width.
\end{compactitem}
\end{theorem}

%Therefore, in the sparse regime (graphs with no $K_{t,t}$) we have equivalence with tree-width and path-width. This partly explains why twin-width is so much reminiscent of tree-width: this is the same approach with only a different way of measuring quality of partitions. A crucial difference however exists in effectivity since we do not have a polynomial algorithm to approximate twin-width. A natural approach is then to find an equivalent notion of twin-width which is more easy to handle. 

This phrases the classic width measures (tree-width, rank-width, path-width, linear rank-width) in the language of twin-width and contraction sequences.
This unifying lens has two main benefits.

The first benefit is simplicity and renewal.
We propose some examples where our characterizations somewhat simplify matters or bring a new, slightly but resolutely different perspective.
We give a short alternative proof of the celebrated theorem by Courcelle, Makowsky, and Rotics~\cite{Courcelle00} that monadic second-order (MSO) model checking (with adjacency relation only) is fixed-parameter tractable on classes with bounded rank-width.
More precisely, $G \models \varphi$ is decidable in time $f(|\varphi|,\text{cw}) |V(G)|$ on graphs $G$ given with a clique-width expression with $\text{cw}$ labels, and MSO sentences $\varphi$.
This is known to generalize Courcelle's theorem that MSO model checking with incidence relation is fixed-parameter tractable on classes with bounded tree-width.

Let us sketch how our algorithm goes.
Instead of parsing a clique-width expression (or a tree-decomposition), we scan the contraction sequence from $G=G_n$ to $G_1$.
We maintain \emph{types} that are ``local to the red graph'', that is, the theory up to quantifier rank $q$ of all the sentences of depth $q$ that are true on a given red component.
As the \local twin-width is bounded, note that there is a bounded number of vertices per red component.
Initially in $G_n=G$, the red components are single vertices, hence the local theory is easy to determine.
Eventually in $G_1$, the whole graph has been merged into a single vertex, thus the local theory of the unique vertex of $G_1$ matches the ``global'' theory of $G$.
As this is precisely what we are after, the crux lies in updating the local theories when moving from trigraph $G_{i+1}$ to trigraph $G_i$.
When contracting $u, v \in V(G_{i+1})$, up to $d+1$ red components of $G_{i+1}$ are fused into one in $G_i$, where $d$ is the upper bound on the \local twin-width.
We show that the local theory on these red components, combined with the black edges sitting on them, is enough to determine unambiguously the local theory of the new red component.
The relative simplicity of \local twin-width brings our proof down to a minimum: one lemma in the vein of the Feferman-Vaught theorem~\cite{Feferman67}.
For completeness, we prove this folklore lemma with Ehrenfeucht-Fraïssé games for MSO.

We also present an analogue to dynamic programming over clique-width expressions or tree-decompositions, with the approach of contraction sequences.
We exemplify it with a practical algorithm for the particular MSO-expressible problem \textsc{$q$-Coloring}.
As the MSO model checking algorithm, it can be described as dynamic programming over the contraction sequence.
However this algorithm is obviously much more practical and simple than the generic MSO model checking.
We then list some advantages that our approach holds over the classic dynamic-programming algorithms. 
In~\cite{twin-width3}, this scheme was used to solve some particular FO-expressible problems on graphs of bounded twin-width given with an $O(1)$-sequence.

This brings us to the second benefit, which is an unexpected collapse of the meta-algorithmic techniques dedicated to handle first-order and monadic second-order logic.
Dynamic programming over contraction sequences tackles in one sweep problems that were seemingly as different and required as disparate techniques as \textsc{$k$-Subgraph Isomorphism} on planar graphs\footnote{See the work of Eppstein~\cite{Eppstein99} building up on Baker's technique~\cite{Baker94} and leading to low tree-width decompositions, and~\cite{twin-width3} for the approach with contraction sequences.} and \textsc{$q$-Coloring} on graphs of bounded rank-width~\cite{complexity}.
Another realization of that collapse is a similar algorithm efficiently solving FO model checking on graphs of bounded twin-width given with $O(1)$-sequences~\cite{twin-width1}, and MSO model checking on graphs of bounded rank-width/bounded \local twin-width. 
This motivates searching for characterizations or generalizations of classes of bounded expansion or nowhere dense classes by means of contraction sequences, as a way to push further their unifying power.

We also define a notion of~\emph{spanning twin-width},\footnote{The exact definition of spanning twin-width is somewhat technical and deferred to~\cref{sec:spanningtww}.} intermediate between bounded tree-width and bounded twin-width, which exactly captures classes excluding a minor, among monotone classes.
\begin{theorem}\label{thm:spanning-tww}
A monotone graph class $\MC$ has bounded spanning twin-width if and only if it is proper minor-closed.
\end{theorem}

So far we explored what happens when \emph{restricting} the notion of bounded twin-width.
Here is an attempt to \emph{generalize} it.
One may observe that homogeneity of two distinct parts $X, Y$ of $\mathcal P_i$ is in fact a directed relation.
Let us say that $Y$ is \emph{homogeneous to $X$} if the pair $X,\{y\}$ is homogeneous for all $y \in Y$.
We can now form the \emph{directed red graph} $D_i$ whose vertices are the parts of $\mathcal P_i$, and with arcs $X \rightarrow Y$ for all pairs $X,Y \in \mathcal P_i$ with $Y$ \emph{not} homogeneous to $X$.
The \emph{oriented width} of $\mathcal P_i$ is the maximum out-degree $w_o$ of $D_i$.
This defines the \emph{oriented twin-width} of $G$ denoted $\otww(G)$ (see~\cref{subfig:otww}).
Note that $\otww(G) \leqslant \tww(G)$ since contracting parts can only create red arcs, which are directed from the contracted vertex.  

Oriented twin-width is in some sense ``fairer'' than twin-width, as far as the ``error count'' is concerned.
When an error (red edge) occurs while contracting two parts, it is only accounted to the newly contracted part, and not to other adjacent parts (since only their red in-degree may increase).
Surprisingly, we will see that:
\begin{theorem}\label{thm:otww-tww}
Oriented twin-width and twin-width are \equi.
\end{theorem}
The proof simply revisits the equivalence between bounded twin-width and so-called \emph{$O(1)$-mixed freeness} (see~\cref{sec:prelim}), and can almost integrally be found in~\cite{twin-width1}.
Our contribution here is mainly conceptual, in identifying the overlooked oriented twin-width.
The indirect nature of the proof, which does not immediately provide an $O(1)$-sequence from a partition sequence with bounded oriented width, suggests that something non-trivial is at play.
Indeed this greatly simplifies, as we then exemplify, the proof that $K_t$-minor free graphs have bounded twin-width presented in~\cite{twin-width1}.
We also observe that planar graphs have oriented twin-width at most~9.
Together with~\cref{thm:otww-tww}, this gives a very direct proof that planar graphs have bounded \emph{twin-width}. 

Another direction to generalize bounded twin-width is to allow contraction sequences to end at ``simple'' (tri)graphs instead of the 1-vertex graph.
Of course for this notion to be new, ``simple'' should not imply bounded twin-width.
Bounded-degree and bounded-expansion are reasonably ``tractable'' classes with unbounded twin-width~\cite{twin-width2}.
We say that a class $\mathcal C$ is \emph{collapsible to a class $\mathcal D$} if graphs of $\mathcal C$ admit \emph{partial} $O(1)$-sequences to (tri)graphs in $\mathcal D$.
We showcase the flexibility of the FO model-checking algorithm in~\cite{twin-width1}: Collapsible classes to bounded degree and collapsible classes to bounded expansion admit respectively a fixed-parameter tractable FO and $\exists$FO model-checking algorithm, provided a corresponding partial $O(1)$-sequence is given. 
This is a relatively elementary fusion of the algorithm in~\cite{twin-width1} and classic techniques from the meta-algorithmic toolbox, namely \emph{Gaifman's locality theorem} and \emph{low tree-depth covers}.
On the one hand, it can be seen as a first attempt to unify and extend tractable FO model-checking algorithms on ``sparse'' classes (bounded degree, bounded expansion) and on possibly ``dense'' classes (bounded twin-width). 
On the other hand, we explain why efficiently finding the corresponding partial $O(1)$-sequences, may turn out simpler than computing (complete) $O(1)$-sequences.

\begin{comment}
\edouard{Do we keep this in the intro? I didn't think too much about it, but there is a case to be made for twin-width to be stable enough that it can stomach the edge orientations but not rank-width. Said otherwise, why would the Marcus-Tardos miracle occur on bounded rank-width?}
An interesting analogy can be drawn.
The \emph{final width} of $D_i$ is the maximum size of a final section of $D_i$ (where the final section generated by a vertex $v$ is the set of vertices reachable by an oriented path starting at $v$). Call \emph{final twin-width} the associated graph width $ftww$, and observe that $ftww(G)\leq ltww(G)$ since the connected components of $G_i$ and $D_i$ are the same but the final sections are (possibly) smaller. Since $ltww$ is equivalent to rank-width, we suspect that the following could hold:

\begin{conjecture}
Rank-width and final twin-width are \equi.
\end{conjecture}
\end{comment}

\paragraph*{Organization of the paper.}
In~\cref{sec:prelim} we recall the relevant background.
In~\cref{sec:classical} we show \cref{thm:cs-width}, present an alternative proof of Courcelle's theorems, and give a practical algorithm for \textsc{$q$-Coloring} on graphs of bounded \local twin-width (i.e., bounded rank-width).
In~\cref{sec:orientedtww}, we prove \cref{thm:otww-tww} and use it to bound the twin-width of $K_t$-minor free graphs.
In~\cref{sec:partial}, we present some FO model-checking algorithms using partial contraction sequences to classes of bounded degree and bounded expansion.
In~\cref{sec:spanningtww}, we show the equivalence between proper minor-closed and bounded spanning twin-width, for monotone classes.

\section{Preliminaries}\label{sec:prelim}

We denote by $[i,j]$ the set of integers $\{i,i+1,\ldots,j-1,j\}$, and $[k]$ is a short-hand for $[1,k]$.
We use the standard graph-theoretic definitions and notations.
Given a graph $G$, its vertex set is denoted by $V(G)$ and its edge set by $E(G)$.
Given a subset $S$ of $V(G)$, $G[S]$ denotes the subgraph of $G$ induced by $S$.

Two parameters $w$ and $w'$ defined on graphs (or more generally on matrices) are \emph{\equi} if there exists a function $f$ such that for every graph (or matrix) $G$, we have $w(G)\leq f(w'(G))$ and $w'(G)\leq f(w(G))$.
Among classical pairs of \equi parameters, let us mention branch-width and tree-width, or rank-width and clique-width.
When speaking of a \emph{class} $\cal C$ of graphs, we mean closed under isomorphism.
When $\cal C$ is furthermore closed under taking subgraphs, we speak of a \emph{monotone class} and when $\cal C$ is closed under induced subgraphs, we speak of a \emph{hereditary class}.
A class $\cal C$ of graphs is \emph{sparse} if there is a $t$ for which no graph in $\cal C$ contains the complete bipartite graph $K_{t,t}$ as a subgraph (i.e., not necessarily induced).

To avoid unnecessary notations, a~\emph{connected component} of a graph may either refer to a (maximal connected) vertex set or its induced subgraph.
More often our connected components will be vertex sets, contrary to the standard definition.

\subsection{Branch decompositions}

A \emph{branch decomposition} of a graph $G$ is a ternary tree $T$ in which the leaves are in one-to-one correspondence with $V(G)$.
In particular every edge $e$ of $T$ corresponds to the vertex bipartition $B_e$ of $G$ defined by the two sets of leaves of the connected components of $T \setminus e$.
Given now any function $m$ from graph bipartitions to the non-negative integers, we obtain a parameter $b_m(G)$, which is the minimum, over all branch-decompositions $T$ of $G$ of the maximum of $m(B_e)$ over all edges $e$ of $T$.

For instance, when $m(B_e)$ is the rank (computed in $\mathbb{F}_2$) of the adjacency matrix of the bipartite subgraph of $G$ spanned by the vertex bipartition $B_e$, the parameter $b_m$ is the \emph{rank-width} of $G$.
We obtain the \emph{linear rank-width} of $G$ by keeping the same rank function $m$ but insisting that branch-decompositions $T$ are ternary trees where the internal vertices form a path. %A class of graphs $\cal C$ can have bounded rank-width, or bounded linear rank-width.
A class has bounded tree-width (resp. bounded path-width) if and only if it is sparse (i.e., $K_{t,t}$-free) and has bounded rank-width (resp. bounded linear rank-width)~\cite{Gurski00}.
Hence tree-width and path-width can be seen as the sparse restrictions of rank-width and linear rank-width.

Let us now introduce a parameter, which is equivalent to rank-width.
The \emph{boolean-width} $\text{bool}(X,Y)$ of a bipartition $(X,Y)$ of $V(G)$ is the logarithm in base 2 of the number of subsets of~$Y$ (equivalently, of~$X$) that are the neighborhood of some subset of~$X$ (resp. of~$Y$).
The \emph{boolean-width} of a graph $G$ is the parameter $b_{\text{bool}}(G)$.
Boolean-width is \equi to rank-width. 
When we only consider branch decompositions in which internal vertices form a path, we speak of \emph{linear boolean-width}, which is similarly equivalent to \emph{linear rank-width}. 

We will use the next observation implicitly, which is proved in~\cite{BUIXUAN20115187}.
\begin{observation}\label{obs:moduleboolean}
  Let $(X,Y)$ be a vertex bipartition of $G$, and let $q$ be the maximum number of vertices in $X$, which have distinct neighborhoods in the bipartite graph $G(X,Y)$ (i.e., ignoring edges within $X$ and within $Y$).
  Then, the boolean-width of $(X,Y)$ is at least $\log_2{q}$ and at most $q$. 
\end{observation}

It is often convenient to root a branch-decomposition $T$ of $G$ at some arbitrary non leaf node.
Given then an internal node $v$ of $T$, the leaves of the subtree of $T$ rooted at $v$ is denoted by $A_v$, and the bipartition $(A_v,V(G)\setminus A_v)$ of $G$ is said to be \emph{associated with $v$}.
Note that there is a unique bipartition associated with an internal node of $T$. 

Branch-decompositions are a very popular concept and can be adapted to partitions of the edge set instead of the vertex set (leading to the branch-width), or to the ground set of a matroid using as parameter its connectivity function (yielding matroid branch-width).
One very appealing feature is that parameters defined by branch-decompositions admit dual parameters via tangles~\cite{OumS06}.
However, the main limitation is that planar graphs have unbounded value for parameters defined with branch-decompositions.
Indeed every balanced bipartition of a planar graph is complex, and thus return a large value for the classical parameters $m$, which in turn gives an unbounded $b_m$.
The main way to overcome the ``planar barrier'' is to measure the complexity of vertex partitions instead of vertex bipartitions.

\subsection{Partition sequences}\label{subsec:part-seq}

A \emph{partition sequence} of a graph $G$ is a sequence ${\mathcal S}=\mathcal P_n, \ldots, \mathcal P_1$ of partitions of $V(G)$ where $\mathcal P_n:=\{\{v\}:v\in V(G)\}$ is the partition of $V(G)$ into singletons, $\mathcal P_1=\{V(G)\}$ is the whole set, and each $\mathcal P_i$ is obtained by merging two parts of $\mathcal P_{i+1}$.
In particular each $\mathcal P_i$ consists of $i$ subsets of $V(G)$.
A \emph{vertex ordering compatible with ${\mathcal S}$} is any total order $\leq$ on $V(G)$ (often also seen as a permutation $\sigma$) such that for every $\mathcal P_i$ and every part $X$ in $\mathcal P_i$, the elements of $X$ are consecutive along $\leq$.

A function $w$ from vertex-partitioned graphs into the non-negative integers is called a \emph{width}.
The \emph{partition-width associated to $w$} of a graph $G$ is the minimum integer $t$ such that there exists a partition sequence $\mathcal P_n, \ldots, \mathcal P_1$ of $G$ such that $w(\mathcal{P}_i) \leqslant t$ for every $i \in [n]$.
We denote it by $p_w(G)$. 

%Ideally, a width should measure how natural a vertex-partition of a graph is, and a good way to do so is to put in the same cluster vertices which look the same. To this end,
We say that two disjoint subsets of vertices $X, Y$ of a graph $G$ are \emph{homogeneous} if there are all edges or no edge between them.
More generally, if $G$ is a binary structure, we insist that any two ordered pairs $(x,y)$ and $(x',y')$ with $x,x'\in X$ and $y,y' \in Y$ induce the same structure.
By extension, $X$ is homogeneous with $X$ when $X$ is a singleton, but is not homogeneous with $X$ when $X$ has at least two elements.
Indeed if $x,x'\in X$, the isomorphism type of $x,x$ and the one of $x,x'$ are not the same. % since equality is neither an edge or a non edge. 

Given a graph $G$ and a partition $\mathcal P_i$, we consider an auxiliary trigraph $G_i$, called \emph{quotient trigraph} and denoted by $G/\mathcal P_i$, with vertices the parts of $\mathcal P_i$, \emph{red edges} all pairs of parts $X,Y$, which are not homogeneous, and \emph{black edges} all pairs of parts $X,Y$ for which $X,Y$ is a complete bipartite graph in $G$.
The \emph{red graph} $\mathcal R(G_i)$ (resp.~\emph{black graph}) of $G_i$ has vertex set $V(G_i)$, and edge set its red edge set $R(G_i)$ (resp.~black edge set $E(G_i)$).
A~\emph{red component} of a trigraph is a connected component in its red graph $\mathcal R(G_i)$, while its \emph{(maximum) red degree} is the (maximum) degree of $\mathcal R(G_i)$.   

By our convention, we add a loop $(X,X)$ to every part $X$, which is not a singleton.
Loops count as degree 1.
Note that a vertex $u$ of $G_i$ corresponds to a subset of vertices of $G$, which we denote by $u(G)$.
It will also be convenient to speak of the \emph{total degree} of $u$ in $G_i$, which is the total number of red and black edges incident to $u$.
The graph obtained from $G_i$ by forgetting the colors (i.e., the graph that is the union of the red and black edges) is called the \emph{total graph} of $G_i$.

%\ste{there is possible confusion with total tww}\edouard{I thought the same but didn't find a better name (I wasn't convinced by underlying graph), and finally assumed that the confusion is unlikely to be a real issue.}

%The red graph of $G_i$ can be directed in order to get a better picture of $\mathcal P_i$.
A \emph{module} in a graph $G$ is a subset of vertices $X$ such that for every $x,x'\in X$ and $y \in V(G)\setminus X$, both $xy$ and $x'y$ induce an edge or both induce a non-edge.
Here we do not impose that $X$ is a \emph{maximal} subset satisfying the latter property. 
More generally, if $G$ is a binary structure, we insist that the ordered pairs $xy$ and $x'y$ induce the same structure.
When $X,Y$ are two disjoint subsets of vertices, we say that \emph{$Y$ is homogeneous to $X$} if $Y$ is a module in $G[X\cup Y]$.
We obtain a directed version of $G_i$ as follows: The \emph{directed trigraph} $D_i$ is obtained from $G_i$ by orienting red edges $XY$ as $X\rightarrow Y$ whenever $X$ is not homogeneous to $Y$ and keeping black edges unchanged.
Note that some red edges of $D_i$ can be directed in both ways.

Given $G$ and a partition $\mathcal P_i$, we define four possible widths: $w_o$ is the maximum red out-degree of $D_i$, $w_d$ is the maximum red degree of $G_i$, $w_c$ is the maximum number of vertices in a red connected component of $G_i$, and $w_t$ is the total number of red edges in $G_i$.
Note that $w_o \leq w_d \leq w_c \leq w_t$.
We now obtain the associated width-parameters:
\begin{compactitem}
	\item The \emph{oriented twin-width} of $G$ as the partition-width $\otww(G)$ associated to $w_o$. 
	\item The \emph{twin-width} of $G$ as the partition-width $\tww(G)$ associated to $w_d$. 
	\item The \emph{\local twin-width} of $G$ as the partition-width $\ltww(G)$ associated to $w_c$. 
	\item The \emph{\total twin-width} of $G$ as the partition-width $\ttww(G)$ associated to $w_t$. 
\end{compactitem}

Observe that $\otww(G) \leq \tww(G) \leq \ltww(G) \leq \ttww(G)$.
Note that there is a slight variation with our original definition of twin-width since loops add one to the degree.
%In particular cographs have now twin-width 1.
For the sake of consistency, we will actually drop the red loops for oriented twin-width and twin-width.
%This is a bit unfortunate since we believe that counting loops is the right choice, but this only changes parameters by one unit.
We will see in~\cref{sec:classical,sec:orientedtww} that we did not create new parameters:
\begin{theorem*} The following parameters are \equi:
	\begin{compactitem}
		\item Twin-width and oriented twin-width.
		\item \Local twin-width and rank-width.
		\item \Total twin-width and linear rank-width.
	\end{compactitem}
\end{theorem*}

It is somewhat comforting to be on charted territory since the choices of the widths $w_o,  w_d, w_c, w_t$ are natural.
The equivalence between oriented twin-width and twin-width is very handy, as we will see in~\cref{sec:orientedtww}. 
%The main remaining questions now are: how to design dual parameters for partition-width parameters? What is the next step after twin-width, ideally capturing classes for which FO-model checking is FPT? Can we find a partition-width which exactly capture minor closed classes (hence sitting between tree-width and twin-width)? 

\subsection{The matrix viewpoint}

Given an ${n\times m}$ matrix $M$, a \emph{row-partition} (resp.~\emph{column-partition}) is a partition of the rows (resp.~columns) of~$M$.
A \emph{$(k,\ell)$-partition} (or simply \emph{partition}) of a matrix $M$ is a pair $({\mathcal R}=\{R_1,\dots ,R_k\}, {\mathcal C}=\{C_1,\dots ,C_\ell\})$ where $\mathcal R$ is a row-partition and $\mathcal C$ is a column-partition.
A \emph{contraction} or \emph{merge} of a partition ${\mathcal P}$ consists of replacing two parts in $\mathcal P$ by their union.
A \emph{contraction} of a partition $({\mathcal R},{\mathcal C})$ of a matrix $M$ is obtained by performing one contraction in the row-partition ${\mathcal R}$ or in the column-partition ${\mathcal C}$.

We distinguish two extreme partitions of an $n \times m$ matrix $M$: the \emph{finest partition} where $({\mathcal R},{\mathcal C})$ have size $n$ and $m$, respectively, and the \emph{coarsest partition} where they both have size one.
A \emph{contraction sequence} of an $n \times m$ matrix $M$ is a sequence of partitions $({\mathcal R}^1,{\mathcal C}^1),\dots ,({\mathcal R}^{n+m-1},{\mathcal C}^{n+m-1})$ where
\begin{compactitem}
	\item $({\mathcal R}^1,{\mathcal C}^1)$ is the finest partition,
	\item $({\mathcal R}^{n+m-1},{\mathcal C}^{n+m-1})$ is the coarsest partition, and
	\item for every $i \in [n+m-2]$, $({\mathcal R}^{i+1},{\mathcal C}^{i+1})$ is a contraction of $({\mathcal R}^i,{\mathcal C}^i)$. 
\end{compactitem}

Given a subset $R$ of rows and a subset $C$ of columns in a matrix $M$, the \emph{zone $R \cap C$} denotes the submatrix of all entries of $M$ at the intersection between a row of $R$ and a column of $C$.
A \emph{zone} of a partition pair $({\mathcal R},{\mathcal C})=(\{R_1, \ldots, R_k\},\{C_1, \ldots, C_\ell\})$ is any $R_i \cap C_j$ for $i \in [k]$ and $j \in [\ell]$.
A zone is \emph{constant} if all its entries are identical.
The \emph{error value} of $R_i$ is the number of non-constant zones among all zones in $\{R_i \cap C_1, \ldots, R_i\cap C_\ell\}$.
We adopt a similar definition for the \emph{error value} of $C_j$.
The \emph{error value} of $({\mathcal R},{\mathcal C})$ is the maximum error value taken over all $R_i$ and~$C_j$. The \emph{twin-width} of a matrix $M$ is the minimum $t$ for which there exists a contraction sequence of $M$ consisting of partitions with error value at most $t$. 
In particular, the twin-width of matrices over finite alphabets extends the scope of twin-width from graphs to arbitrary binary structures over finite signatures.
One may define the \emph{twin-width of a binary structure} as the twin-width of the matrix canonically encoding that structure (by placing at row $x$ column $y$ a letter encoding the isomorphism type of $(x,y)$, i.e., the atomic formulas satisfied by $(x,y)$). 

In a contraction sequence of a matrix $M$, one can always reorder the rows and the columns of $M$ in such a way that all parts of all partitions in the contraction sequence consist of consecutive rows or consecutive columns.
To mark this distinction, a \emph{row-division} is a row-partition where every part consists of consecutive rows; with the analogous definition for \emph{column-division}.
A \emph{$(k,\ell)$-division} (or simply \emph{division}) of a matrix $M$ is a pair $({\mathcal R},{\mathcal C})$ of a row-division and a column-division with respectively $k$ and $\ell$ parts.
%A \emph{fusion} of a division is obtained by contraction of two consecutive parts of ${\mathcal R}$ or of ${\mathcal C}$. \ste{do we need fusions?} Fusions are just contractions preserving divisions.
A \emph{division sequence} is a contraction sequence in which all partitions are divisions.

A matrix $M=(m_{i,j})$ is \emph{vertical} (resp. \emph{horizontal}) if $m_{i,j}=m_{i+1,j}$ (resp. $m_{i,j}=m_{i,j+1}$) for all $i,j$.
Observe that a matrix that is both vertical and horizontal is constant.
We say that $M$ is \emph{mixed} if it is neither vertical nor horizontal.
A crucial remark is that a matrix is mixed if and only if it contains a \emph{corner}, i.e., a contiguous 2-by-2 mixed submatrix.
A \emph{$t$-mixed minor}  in $M$ is a division $({\mathcal R},{\mathcal C})=(\{R_1,\dots ,R_t\},\{C_1,\dots ,C_t\})$ such that every zone $R_i\cap C_j$ is mixed (hence contains a corner).
A matrix without $t$-mixed minor is \emph{$t$-mixed free}. The minimum $t$ for which one can reorder the column and the rows of $M$ to form a $t$-mixed free matrix is called the \emph{mixed value} of $M$. 

\begin{theorem}[\cite{twin-width1}]\label{thm:gridtheorem}
	Twin-width and mixed value are \equi for matrices.
\end{theorem}

Given a graph $G$ and a permutation $\sigma$ of its vertex set, we denote by $Adj_{\sigma}(G)$ the adjacency matrix of $G$ in which the columns and the rows are ordered according to $\sigma$.
As usual, given two vertices $u,v$, the entry $Adj_{\sigma}(G)_{u,v}$ is equal to 1 if $uv$ is an edge and $0$ otherwise.
%Note that a more accurate encoding of the graph should involve a special value for the diagonal values $uu$.
%As this would only change parameters by at most a factor of two, we prefer to simply consider that equality corresponds to no edge.
By extension, we say that the \emph{mixed value} of a graph $G$ is the minimum $t$ for which $Adj_{\sigma}(G)$ is $t$-mixed free, taken over all permutations $\sigma$. The link between mixed value and twin-width for graphs was proved in~\cite{twin-width1}:

\begin{theorem}[\cite{twin-width1}]\label{thm:mixedmatrixgraph}
	Twin-width and mixed value are \equi for graphs.
\end{theorem}

\subsection{Bounded expansion and tree-depth covers}\label{subsec:graph-theory}

We recall some definitions from a paper by Plotkin, Rao, and Smith~\cite{Plotkin94} and from the sparsity program of Nešetřil and Ossona de Mendez~\cite{sparsity}.
One possible way of defining a minor of a graph $G$ is by a collection of disjoint sets $B_1, B_2, \ldots, B_h \subseteq V(G)$, called \emph{branch sets}, such that $G[B_i]$ is connected for all $i \in [h]$.
A~\emph{minor} of $G$ is then any graph $H$, say on vertex set~$[h]$, such that $ij \in E(H)$ implies that there is an edge in $G$ with one endpoint in $B_i$ and the other endpoint in $B_j$.
A \emph{depth-$r$ minor} (also called \emph{$r$-shallow minor}) of a graph $G$ is a minor $H$ of $G$ obtainable in such a way that each branch set induces in $G$ a subgraph with radius at most~$r$.
Let us denote by $\nabla_r(G)$ the set of all the depth-$r$ minors of $G$.
In particular $\nabla_r(G)$ is subgraph-closed.
Given a non-decreasing function $f: \mathbb N \to \mathbb N$, we say that a graph $G$ has \emph{expansion $f$} if for every $r \in \mathbb N$, every graph $H \in \nabla_r(G)$ has (maximum) average degree at most~$f(r)$.
A graph class $\mathcal C$ has \emph{expansion $f$} if all its graphs have expansion $f$, and $\mathcal C$ has \emph{bounded expansion} if it has expansion $f$ for some function $f$.
Note that saying that a graph has bounded expansion is meaningless (they all do, individually) but the fact that, for a specific function $f$, a single graph has expansion~$f$ is meaningful.

The \emph{tree-depth} of a graph $G$ is the minimum integer $k$ such that there is rooted forest $F$ of height $k$ on vertex set $V(G)$ with every edge of $G$ being in an ancestor-descendant relationship in $F$.
Bounded tree-depth is more restrictive than bounded tree-width, so in particular, bounded tree-depth graphs have bounded twin-width.
There is a very useful connection between bounded tree-depth and bounded expansion, in the form of \emph{low tree-depth covers}, or the related low tree-depth decompositions~\cite{Nesetril06,Nesetril15}.
A \emph{low tree-depth cover} with parameters $k, f$ of a graph $G$ is a family of $h = f(k)$ subsets $X_1, \ldots, X_h \subseteq V(G)$ such that, for every $i \in [h]$, $G[X_i]$ has tree-depth at most~$k$, and every subset of $V(G)$ of size at most $k$ is fully included in at least one $X_i$.
A graph class $\mathcal C$ has \emph{low tree-depth covers} if there is a function $f$ depending only on $\mathcal C$ such that for every $G \in \mathcal C$ and integer $k$, $G$ has a low tree-depth cover with parameters $k, f$.

\begin{theorem}[\cite{NesetrilM06,Nesetril15}]\label{thm:td-exp}
  Let $\mathcal C$ be a monotone graph class.
  Then $\mathcal C$ has low tree-depth covers if and only if $\mathcal C$ has bounded expansion.
  Furthermore if $\mathcal C$ has bounded expansion, then there is a function $f$ and an algorithm that given a graph $G \in \mathcal C$ and an integer $k$, outputs a low tree-depth cover of $G$ with parameters $k, f$ in linear time $O_{k,f(k)}(|V(G)|)$.
\end{theorem}

\subsection{Finite model theory}

We recall some relevant background from finite model theory.
We denote by FO$_\tau$ and MSO$_\tau$ the set of first-order, respectively monadic second-order, formulas on signature $\tau$.
In first-order, every variable is interpreted as an element of the universe.
In monadic second-order, a first-order variable is interpreted as an element, while a second-order variable is interpreted as a subset of the universe.
We will mainly consider signatures consisting of unary and binary relation symbols only.
Typically the signature $\tau$ will be one of the following:
\begin{compactitem}
\item $\{E\}$, where $E$ is binary: the language of graphs with possible edge orientations and loops;
\item $\{E,\sim\}$, where $\sim$ is interpreted as an equivalence relation: the language of graphs with an unlabeled partition;
\item $\{E,U_1,\ldots,U_d\}$, where $U_1, \ldots, U_d$ are unary relations interpreted as a partition of the universe: the language of colored graphs, or graphs with a labeled partition;
\item $\{\text{inc}\}$, where $\text{inc}$ is interpreted as a vertex-edge incidence relation of a graph. 
\end{compactitem}
MSO$_{\{E\}}$ is usually denoted by MSO$_1$, and MSO$_{\{\text{inc}\}}$ by MSO$_2$, when a theory interprets $\text{inc}$ as above.
A \emph{sentence} is a formula without free variables.
A \emph{relational $\tau$-structure} $\mathscr A$ on universe $A$ gives an interpretation $R^{\mathscr A} \subseteq A^r$ to every $r$-ary relation symbol $R \in \tau$.
A structure $\mathscr A$ is a \emph{model} of a sentence $\varphi$, denoted by $\mathscr A \models \varphi$, if $\varphi$ holds when interpreted on $\mathscr A$.
We will only consider \emph{finite models} where the \emph{universe} $A$ is a finite set.
The \emph{FO model checking} (resp.~\emph{MSO model checking}) asks given a $\tau$-structure $\mathscr A$ and a sentence $\varphi \in$ FO$_\tau$ (resp. $\varphi \in$ MSO$_\tau$) whether $\mathscr A \models \varphi$ holds.
The fragment $\exists \text{FO}$ (existential first-order logic) consists of the formulas with no universal quantifier and all the negations pushed down to atomic formulas.

The \emph{Gaifman graph} of a $\tau$-structure $\mathscr A$ has vertex set its universe $A$, and edges $ab$ whenever $a$ and $b$ appear in the same relation $R^{\mathscr A}$ for some $R \in \tau$. 
The \emph{quantifier rank} (or \emph{quantifier depth}) of a formula $\varphi$ is the largest number of quantifiers that are nested in $\varphi$.
FO$_\tau[q]$ (resp.~MSO$_\tau[q]$) denotes the set of formulas in FO$_\tau$ (resp.~MSO$_\tau$) with quantifier rank at most~$q$.
When the signature is irrelevant or clear from the context, we may omit it, and simply write FO, MSO, FO$[q]$, MSO$[q]$.

If two finite $\tau$-structures are not isomorphic, then there is a sentence that holds in one but not in the other (for instance the sentence that fully describes the former structure).
However it is very well possible that two non-isomorphic $\tau$-structures satisfy the exact same sentences of FO$[q]$ or MSO$[q]$, for some (finite) integer $q$.
Ehrenfeucht-Fraïssé games characterize exactly when that happens.
Initially the game was defined for first-order logic.
We call it the \emph{EF game} and start with its description.
We will then present its extension \emph{MSO-EF} for monadic second-order.

In the EF game, two players \emph{Spoiler} and \emph{Duplicator} confront each other over two $\tau$-structures $\mathscr A$ and $\mathscr B$.
They play a succession of rounds, when Spoiler wants to show that $\mathscr A$ and $\mathscr B$ are not isomorphic, whereas Duplicator tries to argue the opposite.
The $i$-th round goes like this.
Spoiler chooses a structure $\mathscr A$ or $\mathscr B$, and picks one element in it, say $a_i \in A$ (or~$b_i \in B$).
Duplicator answers by picking an element in the other structure, say $b_i \in B$ (resp.~$a_i \in A$).
If after $q$ rounds, $a_i \mapsto b_i$ (for $i \in [q]$) is still an isomorphism between the induced substructures $(\mathscr A,=)[a_1,\ldots,a_q]$ and $(\mathscr B,=)[b_1,\ldots,b_q]$, we say that Duplicator has \emph{survived $q$ rounds} of the EF game.

We write $\mathscr A \eqfo{q} \mathscr B$ if Duplicator has a strategy such that she can survive (at least) $q$ rounds.
The Ehrenfeucht-Fraïssé theorem states that this is equivalent to $\mathscr A$ and $\mathscr B$ agreeing on all the sentences of FO$_\tau[q]$.

\begin{lemma}[Ehrenfeucht-Fraïssé, see Theorem 3.9 in~\cite{Libkin}]\label{lem:ef}
  Let $\mathscr A$ and $\mathscr B$ be two $\tau$-structures.
  Then, $\mathscr A$ and $\mathscr B$ satisfy the same sentences of FO$_\tau[q]$ if and only if $\mathscr A \eqfo{q} \mathscr B$.
\end{lemma}

The MSO-EF game is similar to the EF-game, but Spoiler can (in each round) alternatively decide to play a subset of $A$ (or a subset of $B$), to which Duplicator answers with a subset of $B$ (resp. of~$A$).
Now after $q$ rounds, a tuple of $e$ elements have been played in both $\mathscr A$ and $\mathscr B$, say $(a_1,\ldots,a_e)$ and $(b_1,\ldots,b_e)$ in this order, as well as a tuple of $s$ sets, say $(A_1,\ldots,A_s)$ in $\mathscr A$ and $(B_1,\ldots,B_s)$ in $\mathscr B$, with $q=e+s$.
Duplicator has survived these~$q$ rounds if $a_i \mapsto b_i$ (for $i \in [e]$) is an isomorphism between $(\mathscr A,=,A_1,\ldots,A_s)[a_1,\ldots,a_e]$ and $(\mathscr B,=,B_1,\ldots,B_s)[b_1,\ldots,b_e]$.
Similarly we write $\mathscr A \eqmso{q} \mathscr B$ if Duplicator has a strategy allowing her to survive (at least) $q$ rounds of the MSO-EF game.
The same characterization holds for MSO and the MSO-EF game.

\begin{lemma}[Ehrenfeucht-Fraïssé for MSO, see Corollary 7.8 in~\cite{Libkin}]\label{lem:ef-mso}
  Let $\mathscr A$ and $\mathscr B$ be two $\tau$-structures.
  Then, $\mathscr A$ and $\mathscr B$ satisfy the same sentences of MSO$_\tau[q]$ if and only if $\mathscr A \eqmso{q} \mathscr B$.
\end{lemma}
  
\section{From branch-decompositions to contraction sequences}\label{sec:classical}

We start this section by showing~\cref{thm:cs-width}, that is, the functional equivalence between boolean-width (equivalently rank-width) and \local twin-width, and between linear boolean-width and \total twin-width. 

\subsection{Classical width parameters as contraction sequences}

%\edouard{maybe somewhere we could reflect that \local twin-width is NLC-width at dusk. We don't see the colors of the vertices, but it does not matter as long as we distinguish the red edges. In general we could elaborate that \local twin-width is a colorless definition of clique-width/NLC-width that is handy algorithmically, contrary to rank-width.}
  
\begin{theorem}\label{thm:local}
Boolean-width and \local twin-width are \equi.
\end{theorem}

\begin{proof}
  Let $G$ be a graph.
  We first show that the \local twin-width of $G$ is bounded in terms of the boolean-width of $G$.
  This was essentially done in the first paper of the series~\cite{twin-width1}, but only sketched and without the explicit notion of \local twin-width.

%Since both rank-width and \local twin-width of $G$ are the maximum value over its connected components, we can assume that $G$ is connected. 

  Let $T$ be a rooted branch-decomposition of $G$ whose leaves are bijectively mapped to $V(G)$, and assume that $T$ has boolean-width at most $d$.
  We make a sequence of contractions $G_n, \ldots, G_{\ell}$ such that the size of any red component in the trigraph sequence is at most $2^{d+1}$ and $G_{\ell}$ has at most $2^{d+1}$ vertices. 
  We may assume $G_n$ has at least $2^{d+1}+1$ vertices, as otherwise we are done. 
  Henceforth, a rooted branch-decomposition $T_j$ for each trigraph $G_j$ on at least $2^{d+1}+1$ vertices is constructed along with the contraction sequence while the following invariant, which clarifies what we actually mean by \emph{branch-decomposition of a trigraph}, is maintained:

\begin{center}
($\clubsuit$) For each node $v$ of $T_j$ with $\abs{A_v}\geq 2^{d}+1$, all edges of $G_j$ crossing $(A_v, V(G_j)\setminus A_v)$ are black, and the boolean-width of the bipartition $(A_v, V(G_j)\setminus A_v)$ is at most $d$.
\end{center} 

The invariant $(\clubsuit)$ clearly holds for $G_n=G$ and $T_n=T$.
%Moreover, if $T_i$ has no proper rooted subtree with at least $2^d+1$ leaves, then contracting an arbitrary pair of vertices of $G_i$ trivially preserves the invariant.
%Furthermore $G_i$ has at most $2^{d+1}$ vertices, hence its \local twin-width is at most $2^{d+1}$ for $G_i$ and all $G_j$ with $j\leq i$.
Let $G_{i+1}$ be a trigraph on $i+1\geq 2^{d+1}+1$ vertices and $T_{i+1}$ be a rooted branch-decomposition for which  $(\clubsuit)$ holds. 
We construct $G_{i}$ and $T_{i}$ satisfying the invariant.

Observe that there exists a node $v$ of $T_{i+1}$ such that $2^d+1\leq \abs{A_v} \leq 2^{d+1}$; a node $v$ such that 
$A_v$ has size at least $2^{d}+1$ and which is furthest from the root meets the condition.
By~\cref{obs:moduleboolean} and the second part of  $(\clubsuit)$ applied to $v$, there are two distinct vertices $x,y$ of $G_{i+1}$, which belong to $A_v$ such that $x,y$ have the same (black) neighborhood in $V(G_{i+1})\setminus A_v$.
Now contract $x,y$ to yield $G_i$.
Let $T_i$ be a branch-decomposition of $G_i$ obtained by deleting~$y$ and identifying the node $x$ to the new vertex of $G_i$ resulting from the contraction of $x$ and~$y$. 
Due to the choice of $x,y$ and the first part of ($\clubsuit$) on $i+1$, the edges between $A_v$ (of the new tree $T_i$) and $V(G_i)\setminus A_v$ are all black.
This means that any bipartition of $V(G_i)$ that can potentially contain a newly created red edge of $G_i$ must be associated with a strict descendant of $v$.
By the definition of $v$, any strict descendant of $v$ has at most $2^d$ leaves (both in $T_{i+1}$ and $T_i$) 
and is thus out of the scope of the invariant $(\clubsuit)$. Therefore, the first part of $(\clubsuit)$ is maintained. 

This also means that for any node $u$ of $T_{i}$, which is not a strict descendant of $v$, the bipartition $(A_u,V(G_i)\setminus A_u)$ associated with $u$ 
is the same as the bipartition associated with $u$ in $T_{i+1}$ after deleting one vertex of $G_{i+1}$, namely $y$. Since 
the boolean-width of a bipartition does not increase after  vertex deletion, we conclude that the second part of $(\clubsuit)$ is maintained as well.
Finally, we observe that the invariant $(\clubsuit)$ indicates that $G$ has \local twin-width at most $2^{d+1}$ since any red component of $G_i$ is included inside some $A_v$ with size at most $2^{d+1}$.

To see the other direction, let $\mathcal P_n, \ldots, \mathcal P_1$ be a partition sequence of $G$ such that every connected component of the red graph $G_i$ has at most $d$ vertices. 
Let $\mathcal P'_i$ be the coarsening of $\mathcal P_i$ such that each part of $\mathcal P'_i$ corresponds to a red component of $G_i$, i.e., is the union of parts of $\mathcal P_i$, which form a red component in $G_i$.
Slightly abusing the notation, we call a part of $\mathcal P'_i$ a red component of $\mathcal P_i$.

Let $T_n$ be a star tree rooted at its center $r$, whose $n$ leaves are bijectively mapped to $V(G)$. 
We will iteratively transform a rooted tree $T_{i+1}$ to $T_{i}$ in a way that mirrors the 
merging of parts in $\mathcal P'_i$. The root $r$ will be unchanged throughout the transformations. 
During iterative transformations we maintain the following invariants: 
\begin{itemize}
\item[(a)] The leaves of each connected component of $T_j- r$ are mapped to each part of $\mathcal P'_j$. 
\item[(b)] The root $r$ has as many children as $\abs{\mathcal P'_j}$ and all other internal nodes have two children.
\item[(c)] For every edge of $T_j$, the associated bipartition has boolean-width at most $2^d$. 
\end{itemize}

The invariants~(a)-(c) clearly hold for $j=n$. Suppose $T_n, \ldots , T_{i+1}$ satisfy the invariants~(a)-(c), with $i \in [1,n-1]$.
Notice that $\mathcal P'_i$ is a coarsening of  $\mathcal P'_{i+1}$ (possibly $\mathcal P'_{i+1}=\mathcal P'_i$), that there is a unique red component $C\in \mathcal P'_i$, obtained as the union of some (possibly one) red components $C_1, \ldots, C_s$ of $\mathcal P'_{i+1}$, and that $\mathcal P'_i\setminus \{C\}=\mathcal P'_{i+1}\setminus \{C_1, \ldots, C_s\}$. 
By the invariant~(a), there are subtrees of $T_{i+1}-r$ whose leaves are mapped to parts $C_1, \ldots, C_s$ of $\mathcal P'_{i+1}$. 
Let $t(C_j)$ be the root of the subtree of $T_{i+1}-r$ corresponding to $C_j$ for $j\in [s]$. 
Now, we construct $T_i$ from $T_{i+1}$ as follows: replace the edges connecting the root~$r$ and $t(C_j)$ for $j\in [s]$ by a subcubic tree rooted at $t(C)$ with~$s$ leaves, whose root $t(C)$ becomes the child of~$r$ and whose~$s$ leaves are identified (arbitrarily) with $t(C_j)$ for $j\in [s]$. 

By the induction hypothesis and the construction of $T_i$, the invariants~(a)-(b) are maintained. 
Furthermore, since $C=\bigcup_{i=1}^s C_i$, the red component $C\in \mathcal P'_{i}$ consists of at most $d$ parts of $\mathcal P_{i}$, thus at most $d+1$ parts of $\mathcal P_{i+1}$, and we have $s\leq d+1$.
%We call the partition $C_1,\ldots , C_s, V(G)\setminus G(C)$ the \emph{prominent partition} of $\P_{i+1}$.
To see that the boolean-width of $T_i$ is at most $2^d$, it suffices to check that for all $I\subseteq [s]$, the boolean-width of the bipartition $(\bigcup_{j\in I}C_j, V(G)\setminus \bigcup_{j\in I}C_j)$ is at most $2^d$. 
For a proper subset $I$ of $[s]$, we know that $\bigcup_{j\in I}C_j$ consists of at most $d$ parts of $\mathcal P_{i+1}$ and each of these parts of $\mathcal P_{i+1}$ has the same neighborhood across the bipartition $(\bigcup_{j\in I}C_j, V(G)\setminus \bigcup_{j\in I}C_j)$ since each $C_j$ is a red component. Hence, the vertex set 
$\bigcup_{j\in I}C_j$ has at most $d$ vertices with distinct neighborhood across the bipartition and thus the boolean-width is at most $2^d$. 
For $I=[s]$, the same argument applies once it is noted that $\bigcup_{j\in I}C_j=C$ consists of at most $d$ parts of $\mathcal P_i$. 

As $\mathcal P'_1=\mathcal P_1=\{V(G)\}$, the invariants (a) and (b) at $i=1$ imply that $T_1$ is 
a subcubic tree whose leaves are bijectively mapped to $V(G)$. With the invariant (c), 
we conclude that $T_1$ and the bijection form a boolean decomposition of width at most $2^d$. 
%
%Now assume that $C$ is a red component of $R_i$ and note that no red edge leaves some part of $R_{i+1}$ contained in $C$ to join some part of $R_{i+1}$ not contained in $C$. In particular some red components $C_1,\dots ,C_s$ of $R_{i+1}$ partition $C$. Note that $s\leq t+1$ otherwise the size of $C$ would exceed $t$. Call $V\setminus C,C_1,\dots ,C_s$ a \emph{basic partition}. Observe that there exists a tree $T$ with leaves corresponding to the vertices of $G$ and such that the partition of leaves corresponding to every internal node of $T$ is a basic partition. The tree $T$ is indeed  constructed from all basic partitions. In particular $T$ has maximum degree $t+2$ and every edge cut of $T$ gives a bipartition of $G$ with boolean-width at most $2^t$. Since the definition of rank-width involving ternary trees can be chosen with arbitrary (but fixed) arity up to some controlled increase of value, we obtain that $G$ has bounded rank width.
\end{proof}

\begin{theorem}\label{thm:linearlocal}
Linear boolean-width and \total twin-width are \equi.
\end{theorem}

\begin{proof}
  Let $G$ be a graph.
  We first bound \total twin-width in terms of linear boolean-width.
Let $T$ be a linear branch-decomposition of $G$ (i.e., in which the internal nodes form a path~$P$) with boolean-width at most $d$. We root $T$ 
at an internal node that is an endpoint of the path~$P$. We follow the same proof as for~\cref{thm:local} and observe that every tree $T_i$ is now a linear branch-decomposition. Indeed, the linearity of $T_i$ implies that there is a unique choice of a minimal rooted subtree of $T_i$ with at least $2^d+1$ leaves. Moreover, the invariant $(\clubsuit)$ of~\cref{thm:local} means that the endpoints of any red edge are restricted to the leaves of this subtree. Note 
that there are at most $2^d+1+ {2^d+1 \choose 2}$ red edges in any $G_i$.

Let $\mathcal P_n, \ldots, \mathcal P_{1}$ be a partition sequence of $G$ achieving \total twin-width at most $d$. 
Let $V_i\subseteq V(G)$ be the union of all non-singleton parts of $\mathcal P_i$. 
As every red edge which is not a loop has at least one endpoint with a loop, at most $d$ parts of $\mathcal P_i$ are incident to red edges. 
Note that $V_{i}\setminus V_{i+1}$ has at most two vertices, and it has two vertices only when we contract two singleton parts. 
Consider now any total order $\prec$ on $V(G)$ such that $u\prec v$ 
%if $u$ was contracted before $v$. 
if there exists $i\in [n]$ with $u\in V_i$ and $v\notin V_i$. 
Indeed, we have $V_{n-1} \prec V_{n-2}\setminus V_{n-1} \prec  \cdots  \prec V_2\setminus V_3 \prec V_1\setminus V_2$. 
Now consider the linear branch decomposition $T$ corresponding to $\prec$ and we argue  that the boolean-width of $T$ is bounded by $d$.
Every bipartition corresponds either to a cut $(V_i,V(G)\setminus V_i)$ or to $(V_i\cup \{v\},V(G)\setminus (V_i\cup \{v\}))$ for some $v\in V_i\setminus V_{i+1}$. 
Observe that each non-singleton part of $\mathcal P_i$ is entirely contained in $V_i$. Moreover, each red edge crossing the cut incident with a non-singleton part $P\in \mathcal P_i$ 
multiplies the number of equivalence classes in $V_i$ by a factor of 2, where each class has the same neighborhood across the cut. This implies 
that the number of equivalence classes in $V_i$ is at most $2^{d-1}+1\leq 2^d$. 
Therefore, with~\cref{obs:moduleboolean} we deduce that  the boolean-width of $T$ is at most $2^d$. 
\end{proof}

Here again, we can find equivalent parameters in the sparse regime:

\begin{theorem} In the class of $K_{t,t}$-free graphs: 
\begin{compactitem}
\item tree-width and \local twin-width are \equi.
\item path-width and \total twin-width are \equi.
\end{compactitem}
\end{theorem}

Note that if we do not add red loops to contracted vertices, linear rank-width and \total twin-width are not equivalent because of cographs.
Indeed, still without red loops, there is always a contraction sequence that produces no red edge at all: iteratively contract two twins.
However cographs have unbounded linear rank-width (this folklore fact can be derived from the unboundedness of linear rank-width for trees~\cite{AdlerK15}).
Thus the addition of loops may seem a bit artificial and even made to force the equivalence.
There is however a good reason for loops: From the adjacency-matrix viewpoint (which is both used for rank-width and twin-width), the main diagonal represents equality, a different predicate than the edge (or non-edge) predicate.
%The second reason is that in the sparse regime, counting loops or not for total twin-width result in \equi parameters. \ste{to check}\ejk{maybe some elaboration on the second reason?}

\subsection{Alternative proof of Courcelle's theorems}\label{subsec:courcelle}

We will give an alternative proof to the celebrated result by Courcelle, Makowsky, Rotics~\cite{Courcelle00} that MSO$_1$ model checking can be solved in linear time on bounded clique-width graphs, given with an $O(1)$-expression.
This generalizes the original Courcelle's theorem~\cite{Courcelle90} that MSO$_2$ model checking can be solved in linear time on bounded tree-width graphs.
%Indeed since bounded tree-width classes coincide with $K_{t,t}$-free bounded clique-width classes~\cite{Gurski00},
Indeed MSO$_2$ is not more expressible than MSO$_1$ on graphs of bounded tree-width~\cite[Theorem 9.37]{Courcelle12}, and there is a linear-time {\FPT} algorithm returning a tree-decomposition of optimal width~\cite{Bodlaender96}.
We observe that there are other alternative proofs to the central result of Courcelle, Makowsky, Rotics; for instance, one based on automata~\cite{Ganian10}, and one game-theoretic~\cite{Langer11}.
Let us also mention that the incoming presentation does \emph{not} follow the exposition of the first-order model checking algorithm in~\cite{twin-width1}, but rather its streamlined revisitation by Gajarský, Pilipczuk, Reidl, and Toru\'nczyk~\cite{fomc-alt}. 

We call \emph{MSO rank-$q$ type} (or \emph{type} for short) any set of sentences
$$\tp{q}(G) := \{\varphi \in \text{MSO}_{\{E\}}[q]~:~ G \models \varphi\}$$
where $G$ is a graph, and we recall, $\text{MSO}_{\{E\}}[q]$ denotes the set of MSO sentences on a signature with a single binary relation $E$, and quantifier rank at most~$q$.
Then $\tp{q}(G)$ is called the MSO rank-$q$ type of $G$, or type of $G$ for short.
%We call \emph{MSO rank-$q$ partitioned type} (or \emph{partitioned type} for short) any set of sentences
%$$\tp{q}(G,\mathcal P^o) := \{\varphi \in \text{MSO}_{\{E,\sim,\prec\}}[q]~:~(G,\mathcal P^o) \models \varphi\}$$
%with $G$ a graph, and $\mathcal P^o$ an ordered vertex-partition of $V(G)$.
%The underlying (unordered) partition $\mathcal P$ of $\mathcal P^o$ is encoded by $\sim$ interpreted as an equivalence relation (or cluster graph).
%The binary symbol $\prec$ is interpreted as a strict linear order on $V(G)$ respecting the equivalence relation of $\sim$.
%Hence the vertices within a part of $\mathcal P$ are consecutive along $\prec$.
%This means that $\prec$ induces a total order on the parts of $\mathcal P$, defining the ordered vertex-partition $\mathcal P^o$.

We fix a positive integer~$d$, upperbounding the \local twin-width on the class we want to tackle.
We call \emph{local MSO rank-$q$ partitioned type} (or \emph{local partitioned type} for short) any set of sentences
$$\ltp{q,d}(G,\mathcal P^o,C) := \{\varphi \in \text{MSO}_{\{E,U_1,\ldots,U_d\}}[q]~:~(G[C],\mathcal P^o[C]) \models \varphi\}$$
where $G$ is a graph with a labeled vertex partition $\mathcal P^o$, and $C \subseteq V(G)$ is the set of initial vertices of $G$ landing in a same red component $D$ with at most~$d$ parts.
The unary relations $U_1, \ldots, U_d$ are interpreted as the labeled partition $\mathcal P^o[C]$ of $G[C]$.
Thus some $U_{d'+1}, U_{d'+2}, \ldots, U_d$ may possibly be empty if $\mathcal P^o[C]$ has $d'<d$ parts.
We use the superscript $o$ (for \textbf{o}rdered) for the labeled partition $\mathcal P^o$ since we will usually fix the labeling by giving an ordering of the parts.
The (unlabeled) vertex partition $\mathcal P$ stands for $\mathcal P^o$ ignoring the labels.
We may sometimes identify $C$ with the red component $D$.

It is not difficult to show that there are only finitely many MSO sentences of quantifier rank~$q$ on finitary signatures, up to logical equivalence (see for instance~\cite[Proposition 7.5]{Libkin}).
Furthermore there is an algorithm (running in time function of $q$ and signature $\tau$ only) that lists all the sentences of depth $q$, up to logical equivalence.  
Therefore the number of (local) MSO rank-$q$ (partitioned) types is bounded by a function of $q$ and $d$ only.
A superset of the (local) MSO rank-$q$ (partitioned) types can be listed in time function of $q$ and $d$ only, by listing all the subsets of sentences of depth $q$.
Instead of deciding $G \models \varphi$ for a particular sentence $\varphi$ with quantifier rank~$q$, we will compute~$\tp{q}(G)$.
By the previous observation, this allows to decide $G \models \varphi$ for \emph{every} sentence $\varphi$ with quantifier rank~$q$ (by simply checking if $\varphi \in \tp{q}(G)$) in constant time, if $q$ and $d$ are treated as a constant.

The algorithm will only compute local partitioned types.
Note that for an $n$-vertex graph~$G$ with a partition sequence $\mathcal P_n, \ldots, \mathcal P_1$, $\ltp{q,d}(G,$ $\mathcal P^o_n = \{\{v\}~:~v \in V(G)\},C)$ is easy to determine, with any labeling $\mathcal P^o_n$ of $\mathcal P_n$, since $G/\mathcal P_n$ is isomorphic to the \emph{graph} $G$ and each possible $C$ is a singleton $\{w\}$ (for some $w \in V(G)$).
Thus these local partitioned types all coincide to the one of the 1-vertex graph with its unique labeled partition.
Furthermore $\ltp{q,d}(G,\mathcal P_1 = \{V(G)\},V(G))$ matches\footnote{Strictly speaking $\ltp{q,d}(G,\mathcal P_1 = \{V(G)\},V(G))$ is a superset of $\tp{q}(G)$, but its projection to sentences ignoring the (trivial) labeled partition is exactly $\tp{q}(G)$.} $\tp{q}(G)$, which is the set we are after.

Thus we only need to compute all the local partitioned types $\left(\ltp{q,d}(G,\mathcal P^o_i,C)\right)_C$ from the knowledge of $\left(\ltp{q,d}(G,\mathcal P^o_{i+1},C')\right)_{C'}$.
We will prove that this is possible since the local partitioned types, the contracted pair of parts $(X,X')$, and the black edges of the quotient trigraph are enough to reconstitute the local partitioned type of the new red component containing $X \cup X'$.   
We show that fact with the characterization via the Ehrenfeucht-Fraïssé game for MSO (see~\cref{lem:ef-mso}).
Recall that given two input structures $\mathscr A, \mathscr B$, Duplicator has a strategy to survive $q$ rounds of the MSO-EF game if and only if $\mathscr A$ and $\mathscr B$ satisfy the same sentences of MSO$[q]$, hence have the same rank-$q$ type.

The subsequent~\cref{lem:merge-game} has a technical and lengthy pre-condition that we chose to state outside the lemma environment for the sake of legibility.
It starts here.
Let $(G^1,(X^1_1,\ldots,X^1_{\ell_1})), \ldots,  (G^k,(X^k_1,\ldots,X^k_{\ell_k}))$ be $k$ graphs $G^j$ given with a labeled partition of size $\ell_j$.
Let $(H^1,(Y^1_1,\ldots,Y^1_{\ell_1})), \ldots,  (H^k,(Y^k_1,\ldots,Y^k_{\ell_k}))$ be such that $$(G^j,(X^j_1,\ldots,X^j_{\ell_j})) \eqmso{q} (H^j,(Y^j_1,\ldots,Y^j_{\ell_j})),~~\text{for every}~j \in [k].$$

Let $(G,\mathcal P^o)$ be a graph with a labeled vertex partition made from the disjoint union $$\bigcup_{j \in [k]} (G^j,(X^j_1, \ldots,X^j_{\ell_j}))$$ with parts labeled by the order $\sigma$, say, $$(X^1_1, \ldots,X^1_{\ell_1}, X^2_1, \ldots,X^2_{\ell_2}, \ldots, X^k_1, \ldots,X^k_{\ell_k}),$$ and adding the biclique between some pairs of parts $X^j_h, X^{j'}_{h'}$ prescribed by a meta-graph $B$ on vertex set $\mathcal P^o$. 

The natural bijection $\iota: X^j_h \mapsto Y^j_h$ (for $j \in [k]$ and $h \in [\ell_j]$) allows to transpose $\sigma$ and $B$ to the union of the $H^j$.
Let $(H,\mathcal Q^o)$ be the graph with a labeled partition made from the disjoint union $\bigcup_{j \in [k]} (H^j,(Y^j_1,\ldots,Y^j_{\ell_j}))$ with the parts labeled along $\iota(\sigma)$, and adding the bicliques prescribed by $\iota(B)$.
Finally we distinguish two parts (the parts to be contracted) $X,X'$ in $\mathcal P^o$, and we distinguish the \emph{homologous} parts $Y := \iota(X), Y' := \iota(X')$ in $\mathcal Q^o$.

\begin{lemma}\label{lem:merge-game}
   $(G,\mathcal P^o,X,X') \eqmso{q} (H,\mathcal Q^o,Y,Y')$.
\end{lemma}
\begin{proof}
  The \emph{global} strategy of Duplicator simply follows the corresponding local strategy if a vertex is played, and if a set $S$ is played, the union of the local answers to each projection of~$S$ on the red components is replied.

  More precisely, if Spoiler plays $x_s \in X^j_h$ (or~$y_s \in Y^j_h$), then Duplicator answers $y_s \in Y^j_h$ (resp.~$x_s \in X^j_h$) accordingly to her local strategy on $(G^j,(X^j_1,\ldots,X^j_{\ell_j})), (H^j,(Y^j_1,\ldots,Y^j_{\ell_j}))$.
  Importantly we know that Duplicator replies a vertex of $Y^j_h$ to a vertex of $X^j_h$ played by Spoiler, since otherwise the \emph{local} unary relation $U_h$ over $G^j$ contradicts the partial isomorphism ensured by $(G^j,(X^j_1,\ldots,X^j_{\ell_j})) \eqmso{q} (H^j,(Y^j_1,\ldots,Y^j_{\ell_j}))$. 
  (Duplicator also remembers that moves $(x_s,y_s)$ have been added to the local $j$-th game, in case more moves are played there.)
  If Spoiler plays a set $S_p \subseteq V(G)$, Duplicator considers all the non-empty sets $S_p \cap V(G^j)$ (for $j \in [k]$) and replies $T_p := \bigcup_j A_j$ where $A_j$ is the local answer to $S_p \cap V(G^j)$.
  Duplicator builds similarly an answer $T_p \subseteq V(H)$ to a move $S_p \subseteq V(G)$ by Spoiler.
  
  Since $(G^j,(X^j_1,\ldots,X^j_{\ell_j})) \eqmso{q} (H^j,(Y^j_1,\ldots,Y^j_{\ell_j}))$, the projection of the mapping $x_s \mapsto y_s$ (for $s$ indexing the vertex moves) onto $G^j, H^j$ is a partial isomorphism between the two corresponding local structures.
  Since there is the same (black) graph $B$ on the parts of $\mathcal P$, as $\iota(B)$ on the parts of $\mathcal Q$, there is an edge in $G$ between $x_s \in X^j_h$ and $x_{s'} \in X^{j'}_{h'}$ if and only if there is an edge in $H$ between $y_s \in X^j_h$ and $x_{s'} \in X^{j'}_{h'}$.
  For every $s$ and $p$, $x_s \in S_p$ if and only if $y_s \in T_p$ otherwise the moves $(x_s,y_s)$ and $(S_p \cap V(G^j),T_p \cap V(H^j))$, played in some order, where $x_s \in V(G^j)$ would make Duplicator lose the local $j$-th game.  
  Finally since the parts $X, X'$ and $Y, Y'$ are homologous (under the bijection $\iota$), $x_s \in X$ (resp.~$x_s \in X'$) if and only if $y_s \in Y$ (resp.~$y_s \in Y'$).
  Otherwise we already observed that $(x_s,y_s)$ would have been a losing pair of moves for Duplicator in the corresponding local game.
  Thus the mapping $x_s \mapsto y_s$ (for $s$ indexing the vertex moves) is a partial isomorphism between $(G,\mathcal P^o,X,X')$ and $(H,\mathcal Q^o,Y,Y')$.
\end{proof}

By~\cref{lem:ef-mso}, we have just established that the local partitioned type of a new red component $C'$ obtained by the merge of two parts $X,X'$ is function of the local partitioned type of every component $C_1, \ldots, C_\ell$ ending up in $C'$ after the contraction, the contracted pair $(X,X')$, and the transversal black edges (bicliques) linking some pairs of parts in two distinct $C_i$'s.

The crucial place where the upper bound~$d$ on the \local twin-width comes into play is in the time that the update from $\left(\ltp{q,d}(G,\mathcal P^o_{i+1},C')\right)_{C'}$ to $\left(\ltp{q,d}(G, \mathcal P^o_i,C)\right)_C$ takes.
Let $Z \in \mathcal P_i$ be the result of the merge of the two parts $X,X' \in \mathcal P_{i+1}$.
Since all the red components have size at most~$d$, the set $Z$ is in a red component with a set~$\mathcal Q$ of at most $d-1$ other parts of~$\mathcal P_i$. 
The black edges in $G/\mathcal P_{i+1}$ on the vertex subset $\mathcal Q' := \{X,X'\} \cup \mathcal Q$, the pair of parts $(X,X')$, and the local partitioned type of each red component within~$\mathcal Q'$, account for less than $2^{{d+1 \choose 2}}{d+1 \choose 2}(d+1)f(q)$ outcomes, for some function~$f$.
Thus the transition table giving the new local partitioned type can be precomputed in time depending only on~$d$ and~$q$.
(The red components of $G/\mathcal P_i$ not containing $Z$ do not need an update.)
Treating~$d$ and~$q$ as constant, the update takes constant time, and the overall algorithm, $O(n)$ time.

If the partition sequence is not given with the input graph, we rely on an algorithm approximating rank-width~\cite{Oum08} to find the sequence.
Our proof looks like the original one by Courcelle, Makowsky, Rotics~\cite{Courcelle00}, except it does not need to use transductions to deal with the label-joins and relabelings of the clique-width expression.
Instead everything is concentrated in~\cref{lem:merge-game}, a statement similar to the Feferman-Vaught theorem~\cite{Feferman67}, which is invoked in~\cite{Courcelle00} to handle the disjoint union of two labeled graphs.

It is known that every MSO transduction of a class of bounded clique-width has itself bounded clique-width~\cite{Courcelle12}.
        This fact alternatively follows from the above algorithm combined with the same arguments as used to show that any FO transduction of a class of bounded twin-width has bounded twin-width~\cite{twin-width1}.  

\medskip

\textbf{Extensions.}
This approach can also be used to the solve optimization versions of MSO.
Say, we want to maximize the size of a set $S \subseteq V(G)$ such that $G \models \psi(S)$ for some formula $\psi(X)$ with one free set variable and quantifier rank~$q$.
We now consider the \emph{local MSO rank-$q$ partitioned (0,1)-type} (same as before but for formulas with one free set variable), or \emph{local partitioned 1-type}, for short:
$$\ltp{q,d}(G,\mathcal P^o,C,S) := \{\varphi(X) \in \text{MSO}_{\{E,U_1,\ldots,U_d\}}[q] : S \subseteq C, (G[C],\mathcal P^o[C]) \models \varphi(S)\}.$$
The number of distinct such types is again upperbounded by a function of $d$ and $q$ only, although $S$ may be any subset on up to $|V(G)|$ vertices.
Thus, any fixed triple $(G,\mathcal P^o,C)$ defines an equivalence relation $S \sim S' \Leftrightarrow \ltp{q,d}(G,\mathcal P^o,C,S) = \ltp{q,d}(G,\mathcal P^o,C,S')$ with $g(d,q)$ classes for some computable function $g$.
For every partitioned graph $(G[C],\mathcal P^o[C])$ induced by a red component, we keep only one representative $S^\tau_{\max} \subseteq C$ per type $\tau = \ltp{q,d}(G,\mathcal P^o,C,S^\tau_{\max})$ such that for every $S \subseteq C$, $S \sim S^\tau_{\max}$ implies $|S| \leqslant |S^\tau_{\max}|$.
That is, for each local partitioned 1-type associated to $(G,\mathcal P^o,C)$, we keep a largest subset of~$C$ realizing that type.

The initialization is still straightforward.
The updates merge maximum-cardinality vertex subsets each realizing some prescribed type, indeed producing a largest vertex subset for the resulting type.
The proof follows again from \cref{lem:merge-game}, where one now adds a unary relation which may arbitrarily overlap with the partition $\mathcal P^o$. 

Finally, any finitary logic for which~\cref{lem:merge-game} still holds, like CMSO (i.e., MSO with additional modular counting predicates, interpreted as ``the size of a set is $p$ modulo $q$'' for some fixed integers $p, q$), also admits a tractable model checking algorithm following the same arguments.  

%\edouard{Mention that closure of MSO-transductions is obtained similarly. Also unification with the rewriting of the FO algorithm for bounded twin-width. This might shrink the meta-algorithm textbook eventually.} \ste{Looks very nice. I didn't check everything, but it looks good}

\subsection{Simpler algorithm for a particular MSO$_1$ problem: \textsc{$q$-Coloring}}\label{subsec:q-coloring}

Like for tree-width and clique-width, one can design more practical algorithms for particular MSO-expressible problems, when the \local twin-width is bounded, still utilizing the viewpoint of contraction sequences.
This gives rise to a different dynamic-programming scheme than the one on tree-decompositions or on clique-width expressions.
It comes naturally \emph{positively-instance driven}~\cite{Tamaki19}, that is, generating only positive subproblems.
This is known to have significantly sped up some exact algorithms, as the computation of tree-width and tree-decompositions (see again the work of Tamaki~\cite{Tamaki19}).
The approach by contraction sequences has other advantages that we will list after we give a particular example.

\begin{theorem}\label{thm:q-coloring}
  Given an $n$-vertex graph $G$ and a contraction sequence $G=G_n, \ldots, G_1$ witnessing that its \local twin-width is at most~$d$, the \textsc{$q$-Coloring} problem can be solved in time $O((2^q-1)^d d^2 n)$.
\end{theorem}

\begin{proof}
If $C \subseteq V(G_i)$ is a red component, that is, a connected component in the red graph of $G_i$, we denote by $C(G)$ the set $\bigcup_{u \in C} u(G)$.
A~\emph{$q$-coloring profile} (or \emph{profile} for short) of $C$ is a function $\gamma: V(C) \to 2^{[q]} \setminus \{\emptyset\}$ such that there is a proper $q$-coloring $c$ of $G[C(G)]$ satisfying, for every $u \in C$, that $c(u(G))=\gamma(u)$.
Thus $\gamma$ gives the exact set of colors used by a (contracted) vertex of the red component.
We will maintain for each red component $C$ the complete set of profiles of $C$.

\medskip

\textbf{Description of the algorithm.} Initially in $G_n$, there are $n$ red components isomorphic to the 1-vertex graph.
Thus for each $u \in V(G)$, we store the set of profiles $\{u \mapsto \{1\}, u \mapsto \{2\}, \ldots, u \mapsto \{q\}\}$.
This corresponds to the $q$ ways a vertex can be colored.
Eventually in $G_1$, if there is a profile for the unique red component (again a 1-vertex graph), it means that $G$ is $q$-colorable.

We shall just update the profiles as the red components evolve.
Let $u, v$ be the vertices contracted into a vertex $z$ when going from $G_{i+1}$ to $G_i$.
Let $C$ be the red component of $G_i$ containing $z$, and $C_1, \ldots C_{d'}$ be the red components in $G_{i+1}$ such that $C = \left( \bigcup_{j \in [d']} C_j \setminus \{u,v\} \right) \cup \{z\}$.
Since $|C| \leqslant d$, it holds that $|\bigcup_{j \in [d']} C_j| \leqslant d+1$, and in particular $d' \leqslant d+1$.
Say that $u \in C_a$ and $v \in C_b$.

The update only consists of computing a set of profiles for $C$ (and destroying the set of profiles of $C_1, \ldots, C_{d'}$).
For every $\gamma_1, \ldots, \gamma_{d'}$ in the profile set of $C_1, \ldots, C_{d'}$, respectively, we check in time $O(d^2)$ is there is a black edge between a pair $x \in C_j, y \in C_{j'}$ with $\gamma_j(x) \cap \gamma_{j'}(y) \neq \emptyset$.
If there is no such edge, we add the corresponding union profile $\gamma$ to the profile set of $C$, i.e., $\gamma(z) = \gamma_a(u) \cup \gamma_b(v)$, and $\gamma(x) = \gamma_j(x)$ if $x \neq z$ and $x \in C_j$.
This finishes the description of the algorithm (see~\cref{fig:3-coloring}).

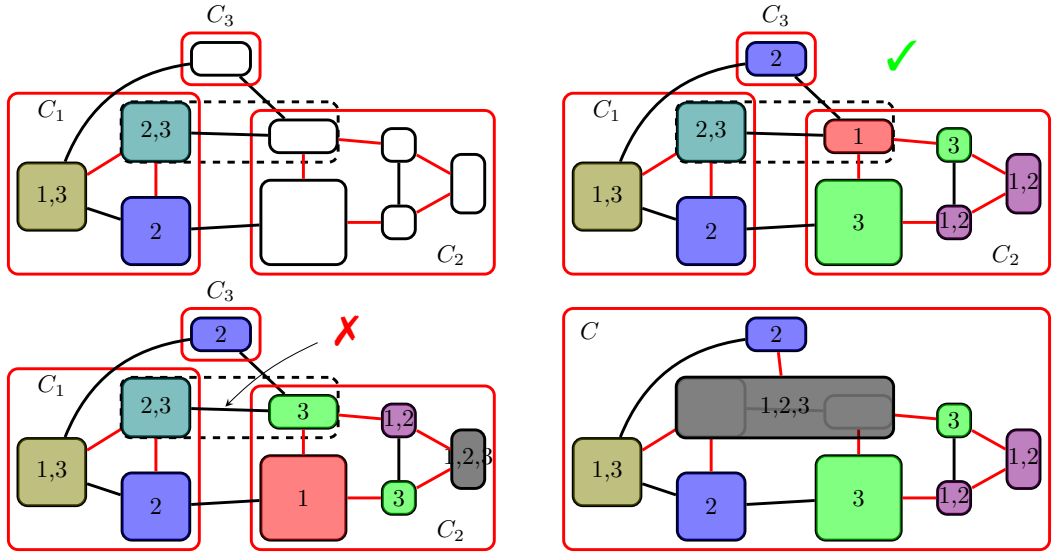
\begin{figure}[h!]
  \centering
  \resizebox{400pt}{!}{
    \begin{tikzpicture}
\def\s{0.5}
\def\op{0.2}

%%% 1 profile
\foreach \i/\j/\l in {0/0/1,1/0/2,0/1/3,1/1/4, 0/3/5,1/3/6,0.5/3.75/7, -2/1/8,-2.5/1/9,-3/1/10,-2/1.5/11,-2.5/1.5/12,-3/1.5/13,-2/2/14,-2.5/2/15,-3/2/16,
 4/0/17,4.75/0/18,5.5/0/19, 4/0.75/20,5.5/0.75/21, 4/1.5/22,4.75/1.5/23,5.5/1.5/24,
 4.25/3.25/25,5.25/3.25/26,
 7.5/0.75/27,7.5/3/28,
 9.5/1.5/29,9.5/2.25/30,
 2/5.5/31,2.75/5.5/32}{
  \node[circle,opacity=\op,inner sep=0.08cm] (v\l) at (\s * \i,\s * \j) {} ;
}

\foreach \i/\l in {{(v1) (v2) (v3) (v4)}/a1,{(v5) (v6) (v7)}/a2,{(v8) (v16)}/a3,{(v17) (v24)}/a4,{(v25) (v26)}/a5,{(v27)}/a6,{(v28)}/a7,{(v29) (v30)}/a8,{(v31) (v32)}/a9}{
\node[draw,very thick, rounded corners, fit=\i] (\l) {} ;
}
\draw[very thick] (a5) -- (a9) ;
\foreach \i/\l in {{(a2) (a5)}/x}{
\node[draw, very thick, dashed, rounded corners, inner sep=-0.008cm, fit=\i] (\l) {} ;
}
\foreach \i/\l in {{(a1) (a2) (a3)}/c1, {(a4) (a5) (a6) (a7) (a8)}/c2, {(a9)}/c3}{
\node[draw, very thick, red, rounded corners, fit=\i] (\l) {} ;
}

\node at (-2.5 * \s, 4 * \s) {$C_1$} ;
\node at (9 * \s, -0.2 * \s) {$C_2$} ;
\node at (2.37 * \s, 6.75 * \s) {$C_3$} ;

\foreach \i/\c/\l in {{(v1) (v2) (v3) (v4)}/blue/2,{(v5) (v6) (v7)}/green!50!blue/{2,3},{(v8) (v16)}/green!50!red/{1,3}}{
\node[fill opacity=0.5,fill=\c,rounded corners, fit=\i] {} ;
\node[rounded corners, fit=\i,label=center:\l] {} ;
}

\foreach \i/\j/\c in {a1/a2/red,a2/a3/red,a1/a3/black,a4/a5/red,a1/a4/black,a2/a5/black,a6/a7/black,a4/a6/red,a5/a7/red,a6/a8/red,a7/a8/red}{
\draw[very thick, \c] (\i) -- (\j) ;
}
\draw[very thick, black] (a3) to [bend left=30] (a9) ;

%%% profiles compatibles
\begin{scope}[xshift=8cm]
\node[circle] (cm) at (3,2.8) {\textcolor{green}{\LARGE{\cmark}}} ;

\foreach \i/\j/\l in {0/0/1,1/0/2,0/1/3,1/1/4, 0/3/5,1/3/6,0.5/3.75/7, -2/1/8,-2.5/1/9,-3/1/10,-2/1.5/11,-2.5/1.5/12,-3/1.5/13,-2/2/14,-2.5/2/15,-3/2/16,
 4/0/17,4.75/0/18,5.5/0/19, 4/0.75/20,5.5/0.75/21, 4/1.5/22,4.75/1.5/23,5.5/1.5/24,
 4.25/3.25/25,5.25/3.25/26,
 7.5/0.75/27,7.5/3/28,
 9.5/1.5/29,9.5/2.25/30,
 2/5.5/31,2.75/5.5/32}{
  \node[circle,opacity=\op,inner sep=0.08cm] (v\l) at (\s * \i,\s * \j) {} ;
}

\foreach \i/\l in {{(v1) (v2) (v3) (v4)}/a1,{(v5) (v6) (v7)}/a2,{(v8) (v16)}/a3,{(v17) (v24)}/a4,{(v25) (v26)}/a5,{(v27)}/a6,{(v28)}/a7,{(v29) (v30)}/a8,{(v31) (v32)}/a9}{
\node[draw,very thick, rounded corners, fit=\i] (\l) {} ;
}
\draw[very thick] (a5) -- (a9) ;
\foreach \i/\l in {{(a2) (a5)}/x}{
\node[draw, very thick, dashed, rounded corners, inner sep=-0.008cm, fit=\i] (\l) {} ;
}
\foreach \i/\l in {{(a1) (a2) (a3)}/c1, {(a4) (a5) (a6) (a7) (a8)}/c2, {(a9)}/c3}{
\node[draw, very thick, red, rounded corners, fit=\i] (\l) {} ;
}

\node at (-2.5 * \s, 4 * \s) {$C_1$} ;
\node at (9 * \s, -0.2 * \s) {$C_2$} ;
\node at (2.37 * \s, 6.75 * \s) {$C_3$} ;

\foreach \i/\c/\l in {{(v1) (v2) (v3) (v4)}/blue/2,{(v5) (v6) (v7)}/green!50!blue/{2,3},{(v8) (v16)}/green!50!red/{1,3},
{(v17) (v24)}/green/3,{(v25)(v26)}/red/1,{(v27)}/blue!50!red/{1,2},{(v28)}/green/3, {(v29) (v30)}/blue!50!red/{1,2}, {(v31) (v32)}/blue/2}{
\node[fill opacity=0.5,fill=\c,rounded corners, fit=\i] {} ;
\node[rounded corners, fit=\i,label=center:\l] {} ;
}

\foreach \i/\j/\c in {a1/a2/red,a2/a3/red,a1/a3/black,a4/a5/red,a1/a4/black,a2/a5/black,a6/a7/black,a4/a6/red,a5/a7/red,a6/a8/red,a7/a8/red}{
\draw[very thick, \c] (\i) -- (\j) ;
}
\draw[very thick, black] (a3) to [bend left=30] (a9) ;
\end{scope}

\begin{scope}[yshift=-4cm]
%%% profiles incompatibles
\node[circle] (cr) at (3,2.8) {\textcolor{red}{\LARGE{\xmark}}} ;
\draw[thin,-stealth] (cr) to [bend left=-10] ++(-1.75,-1.1) ;

\foreach \i/\j/\l in {0/0/1,1/0/2,0/1/3,1/1/4, 0/3/5,1/3/6,0.5/3.75/7, -2/1/8,-2.5/1/9,-3/1/10,-2/1.5/11,-2.5/1.5/12,-3/1.5/13,-2/2/14,-2.5/2/15,-3/2/16,
 4/0/17,4.75/0/18,5.5/0/19, 4/0.75/20,5.5/0.75/21, 4/1.5/22,4.75/1.5/23,5.5/1.5/24,
 4.25/3.25/25,5.25/3.25/26,
 7.5/0.75/27,7.5/3/28,
 9.5/1.5/29,9.5/2.25/30,
 2/5.5/31,2.75/5.5/32}{
  \node[circle,opacity=\op,inner sep=0.08cm] (v\l) at (\s * \i,\s * \j) {} ;
}

\foreach \i/\l in {{(v1) (v2) (v3) (v4)}/a1,{(v5) (v6) (v7)}/a2,{(v8) (v16)}/a3,{(v17) (v24)}/a4,{(v25) (v26)}/a5,{(v27)}/a6,{(v28)}/a7,{(v29) (v30)}/a8,{(v31) (v32)}/a9}{
\node[draw,very thick, rounded corners, fit=\i] (\l) {} ;
}
\draw[very thick] (a5) -- (a9) ;
\foreach \i/\l in {{(a2) (a5)}/x}{
\node[draw, very thick, dashed, rounded corners, inner sep=-0.008cm, fit=\i] (\l) {} ;
}
\foreach \i/\l in {{(a1) (a2) (a3)}/c1, {(a4) (a5) (a6) (a7) (a8)}/c2, {(a9)}/c3}{
\node[draw, very thick, red, rounded corners, fit=\i] (\l) {} ;
}

\node at (-2.5 * \s, 4 * \s) {$C_1$} ;
\node at (9 * \s, -0.2 * \s) {$C_2$} ;
\node at (2.37 * \s, 6.75 * \s) {$C_3$} ;

\foreach \i/\c/\l in {
{(v1) (v2) (v3) (v4)}/blue/2,{(v5) (v6) (v7)}/green!50!blue/{2,3},{(v8) (v16)}/green!50!red/{1,3},
{(v17) (v24)}/red/1,{(v25) (v26)}/green/3,{(v27)}/green/3,{(v28)}/blue!50!red/{1,2}, {(v29) (v30)}/black/{1,2,3}, {(v31) (v32)}/blue/2}{
\node[fill opacity=0.5,fill=\c,rounded corners, fit=\i] {} ;
\node[rounded corners, fit=\i,label=center:\l] {} ;
}

\foreach \i/\j/\c in {a1/a2/red,a2/a3/red,a1/a3/black,a4/a5/red,a1/a4/black,a2/a5/black,a6/a7/black,a4/a6/red,a5/a7/red,a6/a8/red,a7/a8/red}{
\draw[very thick, \c] (\i) -- (\j) ;
}
\draw[very thick, black] (a3) to [bend left=30] (a9) ;
\end{scope}

%%% nouveau profile
\begin{scope}[yshift=-4cm, xshift=8cm]
\foreach \i/\j/\l in {0/0/1,1/0/2,0/1/3,1/1/4, 0/3/5,1/3/6,0.5/3.75/7, -2/1/8,-2.5/1/9,-3/1/10,-2/1.5/11,-2.5/1.5/12,-3/1.5/13,-2/2/14,-2.5/2/15,-3/2/16,
 4/0/17,4.75/0/18,5.5/0/19, 4/0.75/20,5.5/0.75/21, 4/1.5/22,4.75/1.5/23,5.5/1.5/24,
 4.25/3.25/25,5.25/3.25/26,
 7.5/0.75/27,7.5/3/28,
 9.5/1.5/29,9.5/2.25/30,
 2/5.5/31,2.75/5.5/32}{
  \node[circle,opacity=\op,inner sep=0.08cm] (v\l) at (\s * \i,\s * \j) {} ;
}

\foreach \i/\l/\op in {{(v1) (v2) (v3) (v4)}/a1/1,{(v5) (v6) (v7)}/a2/0.2,{(v8) (v16)}/a3/1,{(v17) (v24)}/a4/1,{(v25) (v26)}/a5/0.2,{(v27)}/a6/1,{(v28)}/a7/1,{(v29) (v30)}/a8/1,{(v31) (v32)}/a9/1}{
\node[opacity=\op,draw,very thick, rounded corners, fit=\i] (\l) {} ;
}

\foreach \i/\c/\l in {{(v1) (v2) (v3) (v4)}/blue/2,{(v8) (v16)}/green!50!red/{1,3},
{(v17) (v24)}/green/3,{(v27)}/blue!50!red/{1,2},{(v28)}/green/3, {(v29) (v30)}/blue!50!red/{1,2}, {(v31) (v32)}/blue/2}{
\node[fill opacity=0.5,fill=\c,rounded corners, fit=\i] {} ;
\node[rounded corners, fit=\i,label=center:\l] {} ;
}

\foreach \i/\j/\c in {a1/a2/red,a2/a3/red,a1/a3/black,a4/a5/red,a1/a4/black,a2/a5/{black!20!white},a6/a7/black,a4/a6/red,a5/a7/red,a6/a8/red,a7/a8/red}{
\draw[very thick, \c] (\i) -- (\j) ;
}
\draw[very thick, black] (a3) to [bend left=30] (a9) ;

\node[draw,very thick,fill opacity=0.5,fill=black,rounded corners, inner sep=-0.005cm, fit=(a2) (a5)] (z) {} ;
\node[rounded corners, inner sep=-0.005cm, fit=(a2) (a5),label=center:{1,2,3}] () {} ;
\draw[very thick,red] (z) -- (a9) ;
\node[draw, very thick, red, rounded corners, fit=(a1) (a2) (a3) (a4) (a5) (a6) (a7) (a8) (a9)] {} ;
\node at (-3 * \s, 5.6 * \s) {$C$} ;
\end{scope}
\end{tikzpicture}
  }
  \caption{Illustration of $q$-coloring profiles for $q=3$, with 1 in red, 2 in blue, 3 in green, and update of the profiles when the two vertices in the dashed box are contracted. Example of a $q$-profile of $C_1$ (top-left), ``compatible'' set of profiles of $C_1, C_2, C_3$ (top-right) giving rise to a union profile of~$C$ (bottom-right), and ``incompatible'' set of profiles of $C_1, C_2, C_3$ (bottom-left).
   }
\label{fig:3-coloring}
\end{figure}

\medskip

\textbf{Correctness and running time.}
The correctness comes from the invariant that every red component is associated to its set of profiles.
Indeed a black edge between $xy \in E(G_i)$ means that there is a biclique between $x(G)$ and $y(G)$, thus these sets cannot use a shared color.
The running time is as indicated since $$\prod_{j \in [d']} \#\text{profile}(C_j) \leqslant \prod_{j \in [d']} (2^q-1)^{|C_j|}=(2^q-1)^{\sum\limits_{j \in [d']}|C_j|} = (2^q-1)^{d+1},$$
where $\#\text{profile}(C_j)$ is the number of profiles of red component $C_j$.
To actually compute a coloring, one can simply augment profiles with one representative coloring.
\end{proof}

\textbf{Advantages.} Assuming the SETH,\footnote{For Strong Exponential Time Hypothesis, the assumption that for every $\varepsilon > 0$, there is an integer $k$ such that $n$-variable \textsc{$k$-SAT} cannot be solved in time $(2-\varepsilon)^n$ by a classical algorithm.} this new approach will for instance not improve the theoretically best algorithm for \textsc{$q$-Coloring} parameterized by clique-width, since Lampis showed that running time $O^*((2^q-2)^{\text{cw}})$ is achievable and essentially optimal~\cite{Lampis20}.
However our algorithm presents some practical advantages.

The first remarkable feature is its simplicity.
Contrary to dynamic programming on clique-width expressions, which has to deal with unions, joins, and relabelings (or tree-decompositions with their forget, introduce, and join internal nodes), we have only one operation to handle: the contraction of two vertices, where all optimization efforts can be invested.
We have only $n-1$ operations in total, while tree-decompositions and clique-width parse trees typically have $O(n)$ nodes, incurring a multiplicative overhead.

We do not maintain partial solutions that turn out to be locally infeasible.
When a red component $C$ has at least one profile, we know that $G[C(G)]$ is $q$-colorable.
On the contrary, in the usual algorithm parameterized by clique-width, a join between two labels sharing at least one color can happen long after the corresponding vertices were introduced.
This causes to maintain a lot of unnecessary partial solutions.

\section{Oriented twin-width}\label{sec:orientedtww}

Oriented twin-width is ``fairer'' than twin-width in the following sense: In the partition sequence, when merging two parts $X,Y$ of $\mathcal P_{i+1}$ to form $\mathcal P_i$, the only red arcs of $D_i$ that are created are directed from $X \cup Y$.
Indeed, if $Z$ is homogeneous to $X$ and to $Y$, it is also homogeneous to $X \cup Y$.
Thus any error due to a contraction is only attributed to the contracted vertices and does not wildly spread to their neighbors.
This locality of error makes one's life much easier when designing partition sequences.
Let us illustrate why.

Given a graph $G$ and two non necessarily adjacent vertices $x,y$, we denote by $G/\{x,y\}$ the graph obtained by contracting $x,y$ into one vertex $\{x,y\}$ and joining it to all neighbors of $x$ and $y$ in $G$.
We say that a class $\mathcal C$ of graphs is \emph{$d$-contractible} if for every graph $G$ of $\mathcal C$ there are two vertices $x,y$ such that $G/\{x,y\}$ is also in $\mathcal C$ and is such that the degree of $\{x,y\}$ in the resulting graph is at most $d$.
For instance the following lemma due to Norine et al.~\cite{Norine06} implies that $K_t$-minor free graphs are $2^{\Tilde{O}(t)}$-contractible, where $\Tilde{O}(t) = O(t) \log^k (t)$ for some fixed $k$. 

\begin{lemma}[Lemma 2.2. in \cite{Norine06}]\label{lem:minor-free}
	Let $G$ be a $K_t$-minor free graph, for some integer $t \geq 3$.
	Then there are two vertices $u, v \in V(G)$, both of degree~$2^{\Tilde{O}(t)}$, that are either false twins or adjacent.
\end{lemma}
Moreover, by a direct application of the discharging method, Kotzig~\cite{Kotzig55} could show that planar graphs are $9$-contractible (and the bound is attained by the so-called \emph{stellated icosahedron}).

\begin{lemma}\label{lem:seq-bd-deg}
	For every integer $d$, every $d$-contractible class of graphs $\mathcal C$ has oriented twin-width at most $d$.
\end{lemma}

\begin{proof}
  Let $G$ be a graph on $n$ vertices in $\mathcal C$ and $x,y$ two vertices such that $G/\{x,y\}$ is in $\mathcal C$ and $\{x,y\}$ has degree at most $d$.
  To start the partition sequence, consider $\mathcal P_{n-1}$ consisting of singletons and part $\{x,y\}$.
  Note that the only red arcs created by the contraction stem from $\{x,y\}$, yielding out-degree at most $d$ in $D_{n-1}$ (recall~\cref{subsec:part-seq}).
  We inductively iterate the argument on $G/\{x,y\}$ to form a partition sequence in which every vertex in $D_i$ has out-degree at most~$d$.
\end{proof}

In particular, $K_t$-minor free graphs have oriented twin-width $2^{\Tilde{O}(t)}$, and planar graphs have oriented twin-width at most 9. % (we do not know the exact value here). 

Is bounded oriented twin-width a new notion?
Surprisingly, the answer turns out to be negative.
This is quite fortunate since it allows for more flexibility when looking for contraction sequences: One may just worry about the red out-degree. 
In contrast with the easy arguments above, the original proof that proper minor-closed classes have bounded twin-width~\cite{twin-width1} is quite tedious, involving a carefully chosen depth-first-search tree.
Up to our knowledge, no classic result on minor-closed classes directly implies bounded twin-width.
%Moreover, the known upper bound on the twin-width of planar graphs is very large.

%Unfortunately, the equivalence between twin-width and its oriented version deeply relies on the Marcus-Tardos theorem, which does not therefore provide a clear reason why bounded oriented twin-width implies bounded twin-width.

\begin{theorem}\label{thm:oriented}
	Oriented twin-width and twin-width are \equi.	
\end{theorem}
\begin{proof}
	We already observed that a class with twin-width $d$ has a fortiori oriented twin-width as most $d$.
	Moreover mixed value and twin-width are \equi for graphs by \Cref{thm:mixedmatrixgraph}. Thus to show that 
	\begin{compactenum}
		
		\item\label{it:btww}  twin-width,
		
		\item\label{it:botww} oriented twin-width,
		
		\item\label{it:mf} mixed value
		
	\end{compactenum}
	are all \equi, we just need to argue that bounded oriented twin-width implies bounded mixed value. Actually this is similar to the proof that bounded twin-width implies bounded mixed value, which is done in~\cite{twin-width1}. We reproduce the arguments here for completeness.
	
	We show the contrapositive. Let $G$ be a graph with mixed value greater than $2d+2$, hence such that every adjacency matrix of $G$ has a $2d+2$-mixed minor. Fix a partition sequence ${\mathcal S}=\mathcal{P}_n, \ldots, \mathcal{P}_1$ of $G$. Let $\sigma$ be a vertex ordering compatible with $\mathcal S$. Let $\mathcal D=(\mathcal R=\{R_1, \ldots, R_{2d+2}\},\mathcal C=\{C_1, \ldots, C_{2d+2}\})$ be a $2d+2$-mixed minor of $M := \adj{\sigma}{G}$. By design, the partition sequence $\mathcal S$ defines a symmetric division sequence over $M$ since when merging two subsets of vertices, one can contract (simultaneously) the corresponding columns and the corresponding rows.
	
	Recall that the vertices of the red directed graphs $D_i$ are subsets of vertices of $G$. Let $\ell$ be the maximum index such that a vertex of $D_\ell$ fully contains a part $P$ of $\mathcal D$. Without loss of generality, we may assume that $P$ is a column part, thus $P=C_j$ for some $j \in [2d+2]$. As there is a corner in every cell $M[R_i,C_j]$ with $i \in [2d+2]$, there is in particular at least one \emph{row} $r_i$ in each $R_i$ such that $M[r_i,C_j]$ contains two distinct values.
	In $D_\ell$, the $d+1$ vertices $v_1, v_3, \ldots, v_{2d+1}$ respectively corresponding to rows $r_1, r_3, \ldots, r_{2d+1}$ are all in different parts, except possibly one pair $v_{2h-1}, v_{2h+1}$.
	Indeed as we performed a symmetric division sequence on $M$, and stopped the first time a part of the $2d+2$-mixed minor $\mathcal D$ was contained in a part of the sequence, there is at most one part $R_h$ that is contained in a part of $\mathcal{P}_\ell$.
	(One may observe that $\mathcal D$ need not be symmetric, so $h$ is not necessarily equal to $j$.)
	Thus the vertex of $D_\ell$ corresponding to $C_j$ is the source of at least $d$ red arcs.
	Therefore $G$ has oriented twin-width at least~$d$.
\end{proof}

Note that our proof of~\cref{thm:oriented} shows that $\otww(\mathcal C) \leqslant \tww(\mathcal C) \leqslant \exp(\exp(O(\otww(\mathcal C))))$.
It would be interesting to improve the bound given by the second inequality and/or to complement it by a lower bound. As a consequence, $d$-contractible classes have twin-width~$2^{2^{O(d)}}$, and $K_t$-minor free graphs have twin-width $2^{2^{2^{\Tilde{O}(t)}}}$.

\section{Partial contraction sequences to a target class}\label{sec:partial}

In this section, we present a couple of FO model-checking algorithms based on partial contraction sequences.
It consists of pipelining the algorithm of~\cite{twin-width1} with other elements of the meta-algorithmic toolbox.

\medskip

\textbf{Partial sequences.}
For two non-negative integers $d, \Delta$, let $\cbd{d}{\Delta}$ be the class of graphs admitting a partial $d$-sequence to a trigraph of total degree at most~$\Delta$.
A class $\mathcal C$ is said to be~\emph{collapsible to bounded degree} if there are two integers $d, \Delta$ such that $\cC$ is included in $\cbd{d}{\Delta}$.
For a non-negative integer $d$ and a non-decreasing function $f: \mathbb N \to \mathbb N$, let $\cbe{d}{f}$ be the class of graphs admitting a partial $d$-sequence to a trigraph whose total graph has expansion bounded by $f$.
We refer the reader to~\cref{subsec:graph-theory} for the definition of expansion.
A class $\mathcal C$ is said~\emph{collapsible to bounded expansion} if there is an integer $d$ and a function $f: \mathbb N \to \mathbb N$ such that $\mathcal C$ is included in $\cbe{d}{f}$.
Similarly we may say that a class $\mathcal C$ is \emph{collapsible to class $\mathcal C'$} if there is an integer~$d$ such that every graph $G \in \mathcal C$ has a partial $d$-sequence to a trigraph whose total graph is in $\mathcal C'$.

\medskip

\textbf{The FO model checking algorithm in~\cite{twin-width1}.}
We will not need a full description of the algorithm.
It is enough to recall the following.
Let $G_n, \ldots, G_1$ be a contraction sequence of $G$, and $\mathcal P_n, \ldots, \mathcal P_1$ the corresponding partition sequence.
Like the algorithm presented in this paper for MSO model checking in~\cref{subsec:courcelle}, we maintain the \emph{local theory} of sentences of quantifier rank~$q$ rooted at each vertex $u$ of each trigraph $G_i$ of the sequence.
In the case of bounded \local twin-width, the \emph{local theory} was naturally limited to the red component of $u$.
Now that the red graphs can have arbitrary large components, the \emph{local theory} is limited to vertices at distance less than $3^q$ from $u$ in the red graph of $G_i$.
Since the red degree is assumed to be bounded by~$d$, this represents a set, say, $S_{q,d}(u)$ of less than~$d^{3^q}$ vertices.

In~\cite{twin-width1} the \emph{local theory} is not materialized by types but by a tree (called \emph{reduced morphism-tree}), denoted here by $\mathcal T_{q,d}(u)$, of depth $q$ and total size function of $q$ only, containing all the \emph{possible games} in the partitioned graph $(G,\mathcal P_i)[\bigcup_{v \in S_{q,d}(u)} v(G)]$ up to \emph{equivalent moves}.
More precisely, the root of $\mathcal T_{q,d}(u)$ is labeled by the empty sequence, and every child adds a new vertex of $\bigcup_{v \in S_{q,d}(u)} v(G)$ (new move) to the current sequence (branch from the root to the current node).
At this stage, it is not determined yet if a move is played by $\exists$ or $\forall$ player.
One can define by induction what two~\emph{equivalent moves} are.
At the level of leaves (depth~$q$) two equivalent moves are siblings defining the same induced substructures in $(G,\mathcal P_i)[\bigcup_{v \in S_q(u)} v(G)]$ (with equality).
Then two sibling internal nodes are equivalent if there is a bijection between their children such that the paired children would be equivalent if they had the same parent.

Initially, the tree $\mathcal T_{q,d}(v)$ for each vertex $v \in V(G_n)=V(G)$ is easy to compute: it is a path of length $q$ where the only possible move is $v$. 
The tree $\mathcal T_{q,d}(u)$ when $u$ is the unique vertex of $G_1$ is enough to decide $G \models \varphi$ for every sentence $\varphi$ with quantifier rank~$k$.
Indeed such sentence can be effectively rewritten as a prenex\footnote{A prenex formula or sentence has all its quantifiers as a prefix, followed by a quantifier-free formula.} sentence of depth $q := \text{Tower}(k+\log^* k+3)$~\cite[Theorem 2.2 and inequalities (32)]{Pikhurko11}.
As usual the crux of the algorithm is how to update the trees $\mathcal T_{q,d}(u)$ after one contraction is performed.
Per tree, this takes time function of $q$ and $d$ only, while at most $d^{3^q}$ trees may require an update after one contraction; hence the overall running time of $f(d,q)n$ for some computable function $f$.
We will however not need to detail how the update is done.

\medskip

\textbf{Algorithms based on partial sequences.}
We first observe that we can pipeline the FO model checking algorithm on graphs given with a (complete) $O(1)$-sequence, developed in~\cite{twin-width1}, with Gaifman's locality theorem.
Thus, given a corresponding partial sequence, FO model checking is {\FPT} on collapsible classes to bounded degree.
We recall Gaifman's locality theorem.
An \emph{$r$-local formula} with one free variable is a formula $\phi(x)$ such that $\forall \mathscr A, a$: $\mathscr A \models \phi(a)$ if and only if $\mathscr A[N^r_{\mathscr A}[a]] \models \phi(a)$, where $N^r_{\mathscr A}[a]$ is the $r$-neighborhood of $a$, that is, the set of elements at distance at most $r$ from $a$ in the Gaifman graph of $\mathscr A$.
\begin{theorem}[Gaifman's Locality Theorem~\cite{Gaifman82}]\label{thm:gaifman}
  For every positive integer $q$, there are computable integers $k$ and $r$ such that every FO sentence $\varphi$ with quantifier rank~$q$ is equivalent to a Boolean combination of sentences of the form
  $$\exists x_1 \ldots \exists x_k \underset{1 \leqslant i < j \leqslant k}{\bigland} d(x_i,x_j) > 2r~\land \underset{1 \leqslant i \leqslant k}{\bigland} \phi(x_i)$$ where $\phi(x)$ is an $r$-local formula.  
\end{theorem}
In the previous statement, $d(x_i,x_j) > 2r$ is a short-hand for the fact that there is no path of length at most $2r$ between $x_i$ and $x_j$.

\begin{theorem}\label{thm:seq-to-bd}
  Given $G \in \cbd{d}{\Delta}$ with a partial sequence $G=G_n, \ldots, G_s$ such that $G_s$ has total degree $\Delta$, and a sentence $\varphi \in \text{FO}_{E}[k]$, one can decide whether $G \models \varphi$ in time $f(d,\Delta,k) n$ for some computable function $f$.  
\end{theorem}
\begin{proof}
  We run the algorithm of~\cite{twin-width1} on the partial sequence $G_n, \ldots, G_s$ with a small nuance.
  When the total degree of a vertex $u \in V(G_i)$ becomes at most~$\Delta$, we turn all its black incident edges into red.
  We denote this new trigraph $G'_i$ and proceed to the next contraction on $G'_i$ (not $G_i$).
  By that process, the red degree may exceed $d$ but remains bounded by $d+\Delta$ (in the extreme case when $\Delta$ black edges were turned red).
  We thus maintain trees $\mathcal T_{q,d+\Delta}(u)$, with $q$ function of $k$ as given in~\cref{thm:gaifman}.
  When we reach the trigraph $G'_s$, by design all its edges are red, since its total degree is at most~$\Delta$.
  We may therefore interpret $G'_s$ as a mere graph.
  Up to this point the algorithm takes time $g(d,\Delta,q)(n-s)$ for some computable function $g$.

  To apply Gaifman's locality theorem directly, we adopt the partition viewpoint on the contraction sequence.
  Recall that there is a partial partition sequence $\mathcal P_n, \ldots, \mathcal P_s$ corresponding to the partial trigraph sequence $G_n, \ldots, G_s$.
  We consider the structure $\mathscr A := (G, \mathcal P_s, D := \{ab~:~a \in u(G), b \in v(G), uv \in R(G'_s)\})$.
  We add the ``dummy'' graph~$D$ so that the Gaifman graph of $\mathscr A$ is simply a blow-up of the Gaifman graph of $G'_s$.
  We apply~\cref{thm:gaifman} with $r = 3^q$ on FO$_{E,\sim,D}$ sentences that are \emph{not} using the relation $D$.
  The tree $\mathcal T_{q,d+\Delta}(u)$ for every $u \in V(G'_s)$ allows us to determine every such $r$-local sentence $\phi(x)$ of quantifier rank at most~$q$.
  We can therefore mark the vertices $u \in V(G'_s)$ such that $\phi(a)$ holds for at least one vertex $a \in u(G)$.
  This takes time linear in $s$.
  We conclude as in the FO model-checking algorithm for bounded-degree structures of Seese~\cite{Seese96}, by observing that finding a $2r$-scattered set of size $k$ (i.e., $k$ vertices pairwise at distance more than $2r$) among the marked vertices in the graph $G'_s$ can be done in time $h(k,r)s$ for some computable function $h$ (with a bounded search tree).
  Hence the overall running time.
\end{proof}

Typical graphs collapsible to bounded degree --but of unbounded twin-width and unbounded degree-- are blow-ups (replace every vertex by a clique module of arbitrary size) of bounded-degree graphs; and more generally any modular decomposition where all the modules have bounded twin-width, while the core has bounded degree.
We believe that~\cref{thm:seq-to-bd} should hold more generally for collapsible classes to bounded expansion.
We will only show the result for \emph{existential} first-order sentences.

\begin{theorem}\label{thm:seq-to-be}
  Given $G \in \cbe{d}{g}$ with a partial sequence $G=G_n, \ldots, G_s$ such that the total graph of $G_s$ has expansion $g$, and a sentence $\varphi = \exists x_1 \exists x_2 \ldots \exists x_q \psi \in \exists \text{FO}_{E}[q]$ with $\psi$ a quantifier-free formula, one can decide whether $G \models \varphi$ in time $f(d,g,q) n$ for some computable function $f$.  
\end{theorem}
\begin{proof}
  We run the algorithm of~\cite{twin-width1} on the partial sequence $G_n, \ldots, G_s$.
  When we reach $G_s$, we compute a \emph{low tree-depth cover} $X_1, \ldots, X_h$ of the total graph $G'_s$ of $G_s$ with parameters $h, f=f(\cbe{d}{g})$ such that, we recall, $G'_s[X_j]$ has tree-depth at most $q$ for every $j \in [h]$, every subset of $V(G'_s)$ of size at most $q$ is fully contained in at least one $X_j$, and $h=f(k)$~\cref{thm:td-exp}.
  As $G'_s[X_j]$ has tree-depth at most~$q$ for every $j \in [h]$, it has twin-width bounded by a function of~$q$, and a (complete) $f'(q)$-sequence can be found in polynomial time for some function $f'$.
  For every $j \in [h]$, we perform the following~\emph{run}.
  We trim all the trees $\mathcal T_{q,d}(u)$ with $u \in V(G_s)$ by deleting every move that is not in $X_j$ (and its subtree).
  We resume the algorithm of~\cite{twin-width2} with the \emph{trigraph} $G_s[X_j]$ and the twin-width bound set to $f'(q)$.

  If $G \models \varphi$ indeed holds, let $X_j$ be such that $G'_s[X_j]$ contains $q$ vertices $(v_1, \ldots, v_q)$ such that $G \models \psi(v_1, \ldots, v_q)$ (recall that $\varphi$ is existential).
  The corresponding runs detects a solution.
  If $G \models \varphi$ does not hold, every run is negative.
  The claimed overall running time is easy to derive.
\end{proof}

\cref{thm:seq-to-bd} tackles more general graph classes than the FO model-checking algorithm of~\cite{twin-width2}, combining the features of graphs with bounded twin-width and of sparse graphs.
The interest of such algorithms might also be in easing the computation of the (partial) sequence.
It is still unknown if there is an approximation algorithm outputting $f(d)$-sequences for graphs of twin-width at most~$d$, in say, fixed-parameter time.  
A reason why this may be a delicate issue is that $\Omega(\log n)$-subdivisions of $n$-vertex graphs have bounded twin-width, while $o(\log n)$-subdivisions do not~\cite{twin-width2}.
However $\Omega(\log n)$-subdivisions are trivially collapsible to bounded expansion, and easily shown collapsible to bounded degree. 

\section{Spanning twin-width}\label{sec:spanningtww}
 
Let $\preceq$ be a partial order on a set $X$. When $x\preceq y$, we say that $x$ is an \emph{ancestor} of $y$, that $y$ is a \emph{descendant} of $x$, and that $x,y$ are \emph{comparable}. 
A \emph{forest order} $\prec$ on $X$ is a partial order such that whenever $x,y\preceq z$, then $x,y$ are comparable. A \emph{tree order} is a forest order with a minimum element. We write $x\prec^c y$ when $x\prec y$ and there is no $z$ such that $x\prec z\prec y$. The binary relation $(X,\prec^c)$ is the \emph{Hasse diagram} of $\preceq$.
Notice that the Hasse diagram of a tree order is a directed tree. 
 
Let $G$ be a connected graph.
A tree order $\preceq$ on $V(G)$ is \emph{compatible with $G$} if $uv$ is an edge of $G$ whenever $u\prec^c v$.
Put in another way, a tree order compatible with $G$ can be seen as the transitive closure of some oriented rooted spanning tree of $G$.
 
Adding a tree order to a graph can help to design partition sequences.
For instance, the key ingredient in the proof of~\cite{twin-width1} that minor-closed classes have bounded twin-width lies in the fact that if $G$ is $K_t$-minor free, then it has a tree order $\preceq$ (a kind of Lex-DFS) such that $(G,\preceq)$ seen as a binary structure has bounded twin-width.
This is one appealing feature of twin-width: One often has to guess which additional information will guide the sequence.
Here the tree order can be seen as an intermediate step between the mere graph $G$ and the full partition sequence.
 
In order to refine the landscape between bounded tree-width and bounded twin-width, a natural candidate is the \emph{spanning twin-width} of a connected graph $G$, defined as the minimum twin-width of $(G,\preceq)$ taken
over all tree orders $\preceq$ compatible with $G$.
We extend the notion to disconnected graphs by taking the maximum over all connected components.
Observe that this parameter is not monotone.
Indeed the spanning twin-width of a subgraph (even induced) can be larger than the one of the host graph.
Nevertheless for monotone classes, it exactly captures classes excluding a minor.
 
\begin{theorem} \label{thm:monotoneSparseStwwMinor}
A monotone graph class $\MC$ has bounded spanning twin-width if and only if there exists some $t$ such that no graph in $\MC$ admits $K_t$ as a minor. 
\end{theorem}
\begin{proof}
As already mentioned, the backward implication is proved in~\cite{twin-width1} using a Lex-DFS. To show the forward direction, we exhibit graphs in $\MC$ with arbitrarily high spanning twin-width. Here an \emph{induced subdivision} of a graph $H$ is any graph obtained from arbitrarily subdividing the edges of $H$ (including the possibility of not subdividing the edge and just keeping it).
 
\begin{claim}
If a monotone class $\MC$ contains arbitrarily large clique minors, then $\MC$ contains an induced subdivision of every cubic graph.
\end{claim}
\begin{proofofclaim}
		%If $\MC$ is not proper, that is either finite or there are finitely many graphs outside of $\MC$, then it contains all cubic graphs except for finitely many 
		%and the claim follows. 
		%Hence we assume that $\MC$ is proper. 
		%Suppose the contrary, i.e. there is an infinite family $\{H_i\}_{i\in \mathbb{N}}$ of cubic graphs, with 
		%each $H_i$ on $n_i$ vertices and $\{n_i\}_{i\in \mathbb{N}}$ forming a strictly increasing sequence, 
		%such that $\MC$ does not contain any induced subdivision of $H_i$ for all $i$. 
		To see this, observe that if $H$ is a cubic graph on $t$ vertices and $G$ contains a $K_t$ minor, then $G$ contains a minor that contains $H$ as a subgraph.
		Consider then an edge-minimal subgraph $G'$ of $G$ that contains a minor, itself containing $H$ as a subgraph and observe now that $G'$ is exactly an induced subdivision of $H$.	
\end{proofofclaim}

We will now rely on the notion of first-order (FO) transductions.
In general, a transduction transforms a class of structures into a class of structures, applying some rules definable in first-order logic. 
Here we only need a special case of transductions so we give a simplified definition.
A (non-copying) \emph{FO transduction} $\textsc{T}$ of a graph $G$ in a relational structure $\mathcal A$ consists of augmenting $\mathcal A$ with a constant number of unary relations, and then defining the vertex set and edge set of $G$ by means of two first-order formulas with one and two free variables, respectively, in the vocabulary of $\mathcal A$ augmented with the unary relations.
Since the interpretation of the added unary relations can be arbitrary, $\textsc{T}(\mathcal A)$ is not a single graph $G$ but a set of them.
A graph class $\mathcal C$ is a transduction in a class $\mathcal S$ of relational structures if $\mathcal C \subseteq \bigcup_{\mathcal A \in \mathcal S} \textsc{T}(\mathcal A)$. 
Importantly bounded twin-width is preserved by applying FO transductions~\cite{twin-width1}: If $\mathcal C$ is a transduction in a class of binary structures with bounded twin-width, then $\mathcal C$ itself has bounded twin-width.

Our strategy is to show that if an induced subdivision $S$ of a (connected) cubic graph $H$ is given together with a tree order $\preceq$, then one can retrieve $H$ from $(S,\preceq)$ using a fixed FO transduction.
Since FO transductions keep twin-width bounded~\cite{twin-width1}, the assumption that $\MC$ has bounded spanning twin-width would directly imply that the class of cubic graphs has also bounded twin-width, which is false~\cite{twin-width2}.
 
We only have to show how to recover $H$ from $(S,\preceq)$.
FO transductions allow to introduce a bounded number of unary relations, so we can first identify the set $V$ of vertices of $S$ that have degree~3 (which are the original vertices of $H$).
The technical task is now to decide if two such vertices $x,y\in V$ are joined or not in $H$.
This is the case if and only if there is an induced path $x=u_0, \ldots, u_{\ell}=y$ of $S$ such that all internal vertices are not in $V$.
We denote this path by $P$.
Note that when $\ell=1$ querying the existence of such a path can be directly done by asking if $xy$ is an edge  of $S$.
%However arbitrary long paths are usually out of the scope of first order.
Let us focus on the case when we want to retrieve an edge $xy$, which has been subdivided at least once, i.e., let us see how the tree order can help to retrieve the edge $xy$ when $\ell>1$.
 
As $\preceq$ is a tree order compatible with $S$, $xy$ is an edge of $H$ that is subdivided if and only if one of the following holds for $P$:
 
	\begin{enumerate}
		\item $u_0 \prec^c  \cdots \prec^c u_{\ell},$
		\item $u_0 \succ^c  \cdots \succ^c u_{\ell},$ or
		\item there exists $i\in [\ell-1]$ such that $u_0 \prec^c \cdots  \prec^cu_i$ and $u_{i+1}\succ^c \cdots\succ^c u_{\ell}$.
		\item there exists $i\in [\ell-1]$ such that $u_i \prec^c \cdots  \prec^cu_0$ and $u_i \prec^c \cdots  \prec^c u_{\ell}$ \\(exceptional case when $u_i$ is the root of the tree-order).
		\item there exists $i\in [\ell-1]$ such that $u_i \prec^c \cdots  \prec^cu_0$ and $u_{\ell} \prec^c \cdots  \prec^c u_{i+1}$ \\(exceptional case when $u_i$ is the root of the tree-order).
		\item there exists $i\in [\ell-1]$ such that $u_i \prec^c \cdots  \prec^c u_{\ell}$ and $u_{0} \prec^c \cdots  \prec^c u_{i-1}$ \\(exceptional case when $u_i$ is the root of the tree-order).
	\end{enumerate}
 
Since all these conditions can be tested with a (long but bounded) first-order formula, $H$~is a first-order transduction of $(S,\prec)$.
\end{proof}
 
An interesting direction would be to investigate which hereditary classes have bounded spanning twin-width, as it could indicate some possible generalization of minor-closed classes to the dense setting.
But even sparse hereditary classes $\cal C$ with bounded spanning twin-width are somewhat mysterious.
For instance, the induced subgraphs of the grid with added diagonals (a non-planar graph), where two edges are added in each facial $C_4$ to form a $K_4$, have bounded spanning twin-width (and arbitrarily large clique minors). 
 
We believe that subdivisions of arbitrary cubic graphs could play an important role in this study.
More specifically, it is plausible that every hereditary class of graphs with girth at least 5 and no induced subdivisions of a fixed subcubic graph has bounded twin-width (and perhaps even, bounded spanning twin-width).
For example, segment intersection graphs (which exclude induced subdivisions of the 1-subdivision of $K_{3,3}$) without $K_{t,t}$ subgraph were shown to have bounded twin-width~\cite{twin-width8}.

%In that direction, we suggest the following problem as a first step, where a \emph{segment graph} is the intersection graph of segments in the plane.
%\begin{conjecture}
%The class of segment graphs with girth at least five has bounded spanning twin-width.
%\end{conjecture}
%We do not even know if this class has bounded twin-width.

%\bibliography{biblio}

\begin{thebibliography}{10}

\bibitem{AdlerK15}
Isolde Adler and Mamadou~Moustapha Kant{\'{e}}.
\newblock Linear rank-width and linear clique-width of trees.
\newblock {\em Theor. Comput. Sci.}, 589:87--98, 2015.
\newblock \href {https://doi.org/10.1016/j.tcs.2015.04.021}
  {\path{doi:10.1016/j.tcs.2015.04.021}}.

\bibitem{Baker94}
Brenda~S. Baker.
\newblock Approximation algorithms for {NP}-complete problems on planar graphs.
\newblock {\em J. {ACM}}, 41(1):153--180, 1994.
\newblock \href {https://doi.org/10.1145/174644.174650}
  {\path{doi:10.1145/174644.174650}}.

\bibitem{Bodlaender96}
Hans~L. Bodlaender.
\newblock A linear-time algorithm for finding tree-decompositions of small
  treewidth.
\newblock {\em {SIAM} J. Comput.}, 25(6):1305--1317, 1996.
\newblock \href {https://doi.org/10.1137/S0097539793251219}
  {\path{doi:10.1137/S0097539793251219}}.

\bibitem{twin-width8}
{\'{E}}douard Bonnet, Dibyayan Chakraborty, Eun~Jung Kim, Noleen K{\"{o}}hler,
  Raul Lopes, and St{\'{e}}phan Thomass{\'{e}}.
\newblock Twin-width {VIII:} delineation and win-wins.
\newblock {\em CoRR}, abs/2204.00722, 2022.
\newblock \href {http://arxiv.org/abs/2204.00722} {\path{arXiv:2204.00722}},
  \href {https://doi.org/10.48550/arXiv.2204.00722}
  {\path{doi:10.48550/arXiv.2204.00722}}.
  
\bibitem{twin-width2}
{\'{E}}douard Bonnet, Colin Geniet, Eun~Jung Kim, St{\'{e}}phan Thomass{\'{e}},
  and R{\'{e}}mi Watrigant.
\newblock Twin-width {II:} small classes.
\newblock In {\em Proceedings of the 2021 ACM-SIAM Symposium on Discrete
  Algorithms (SODA)}, pages 1977--1996, 2021.
\newblock \href {https://doi.org/10.1137/1.9781611976465.118}
  {\path{doi:10.1137/1.9781611976465.118}}.

\bibitem{twin-width3}
{\'{E}}douard Bonnet, Colin Geniet, Eun~Jung Kim, St{\'{e}}phan Thomass{\'{e}},
  and R{\'{e}}mi Watrigant.
\newblock {Twin-width III: Max Independent Set, Min Dominating Set, and
  Coloring}.
\newblock In Nikhil Bansal, Emanuela Merelli, and James Worrell, editors, {\em
  48th International Colloquium on Automata, Languages, and Programming,
  {ICALP} 2021, July 12-16, 2021, Glasgow, Scotland (Virtual Conference)},
  volume 198 of {\em LIPIcs}, pages 35:1--35:20. Schloss Dagstuhl -
  Leibniz-Zentrum f{\"{u}}r Informatik, 2021.
\newblock \href {https://doi.org/10.4230/LIPIcs.ICALP.2021.35}
  {\path{doi:10.4230/LIPIcs.ICALP.2021.35}}.

\bibitem{twin-width4}
{\'{E}}douard Bonnet, Ugo Giocanti, Patrice~Ossona de~Mendez, Pierre Simon,
  St{\'{e}}phan Thomass{\'{e}}, and Szymon Toru\'{n}czyk.
\newblock Twin-width {IV:} ordered graphs and matrices.
\newblock {\em CoRR}, abs/2102.03117, accepted at STOC 2022.
\newblock URL: \url{https://arxiv.org/abs/2102.03117}, \href
  {http://arxiv.org/abs/2102.03117} {\path{arXiv:2102.03117}}.

\bibitem{twin-width1}
{\'{E}}douard Bonnet, Eun~Jung Kim, St{\'{e}}phan Thomass{\'{e}}, and
  R{\'{e}}mi Watrigant.
\newblock Twin-width {I:} tractable {FO} model checking.
\newblock {\em J. {ACM}}, 69(1):3:1--3:46, 2022.
\newblock \href {https://doi.org/10.1145/3486655} {\path{doi:10.1145/3486655}}.

\bibitem{twin-width&permutations}
{\'{E}}douard Bonnet, Jaroslav Nešetřil, Patrice~Ossona de~Mendez, Sebastian
  Siebertz, and St{\'{e}}phan Thomass{\'{e}}.
\newblock Twin-width and permutations.
\newblock {\em CoRR}, abs/2102.06880, 2021.
\newblock URL: \url{https://arxiv.org/abs/2102.06880}, \href
  {http://arxiv.org/abs/2102.06880} {\path{arXiv:2102.06880}}.

\bibitem{BUIXUAN20115187}
Binh-Minh Bui-Xuan, Jan~Arne Telle, and Martin Vatshelle.
\newblock Boolean-width of graphs.
\newblock {\em Theoretical Computer Science}, 412(39):5187--5204, 2011.
\newblock URL:
  \url{https://www.sciencedirect.com/science/article/pii/S030439751100418X},
  \href {https://doi.org/https://doi.org/10.1016/j.tcs.2011.05.022}
  {\path{doi:https://doi.org/10.1016/j.tcs.2011.05.022}}.

\bibitem{Courcelle90}
Bruno Courcelle.
\newblock {The monadic second-order logic of graphs. I. Recognizable sets of
  finite graphs}.
\newblock {\em Information and Computation}, 85(1):12 -- 75, 1990.
\newblock URL:
  \url{http://www.sciencedirect.com/science/article/pii/089054019090043H},
  \href {https://doi.org/https://doi.org/10.1016/0890-5401(90)90043-H}
  {\path{doi:https://doi.org/10.1016/0890-5401(90)90043-H}}.

\bibitem{Courcelle12}
Bruno Courcelle and Joost Engelfriet.
\newblock {\em {Graph Structure and Monadic Second-Order Logic - A
  Language-Theoretic Approach}}, volume 138 of {\em Encyclopedia of mathematics
  and its applications}.
\newblock Cambridge University Press, 2012.
\newblock URL:
  \url{http://www.cambridge.org/fr/knowledge/isbn/item5758776/?site\_locale=fr\_FR}.

\bibitem{Courcelle00}
Bruno Courcelle, Johann~A. Makowsky, and Udi Rotics.
\newblock Linear time solvable optimization problems on graphs of bounded
  clique-width.
\newblock {\em Theory Comput. Syst.}, 33(2):125--150, 2000.
\newblock \href {https://doi.org/10.1007/s002249910009}
  {\path{doi:10.1007/s002249910009}}.

\bibitem{complexity}
Marek Cygan, Fedor~V Fomin, {\L}ukasz Kowalik, Daniel Lokshtanov, D{\'a}niel
  Marx, Marcin Pilipczuk, Micha{\l} Pilipczuk, and Saket Saurabh.
\newblock {\em Parameterized algorithms}, volume~4.
\newblock Springer, 2015.

\bibitem{Eppstein99}
David Eppstein.
\newblock Subgraph isomorphism in planar graphs and related problems.
\newblock {\em J. Graph Algorithms Appl.}, 3(3):1--27, 1999.
\newblock \href {https://doi.org/10.7155/jgaa.00014}
  {\path{doi:10.7155/jgaa.00014}}.

\bibitem{Feferman67}
Solomon Feferman and Robert~L Vaught.
\newblock The first order properties of products of algebraic systems.
\newblock {\em Journal of Symbolic Logic}, 32(2), 1967.

\bibitem{Gaifman82}
Haim Gaifman.
\newblock On local and non-local properties.
\newblock In {\em Studies in Logic and the Foundations of Mathematics}, volume
  107, pages 105--135. Elsevier, 1982.

\bibitem{fomc-alt}
Jakub Gajarsk{\'{y}}, Michal Pilipczuk, Felix Reidl, and Szymon Toru\'nczyk.
\newblock personal communication, also
  \url{https://www.youtube.com/watch?v=BddeeHWRX_g&ab_channel=FAS} starting at
  1:59:10.
\newblock 2021.

\bibitem{Ganian10}
Robert Ganian and Petr Hlinen{\'{y}}.
\newblock On parse trees and myhill-nerode-type tools for handling graphs of
  bounded rank-width.
\newblock {\em Discret. Appl. Math.}, 158(7):851--867, 2010.
\newblock \href {https://doi.org/10.1016/j.dam.2009.10.018}
  {\path{doi:10.1016/j.dam.2009.10.018}}.

\bibitem{Gurski00}
Frank Gurski and Egon Wanke.
\newblock The tree-width of clique-width bounded graphs without
  \emph{K\({}_{\mbox{n, n}}\)}.
\newblock In Ulrik Brandes and Dorothea Wagner, editors, {\em Graph-Theoretic
  Concepts in Computer Science, 26th International Workshop, {WG} 2000,
  Konstanz, Germany, June 15-17, 2000, Proceedings}, volume 1928 of {\em
  Lecture Notes in Computer Science}, pages 196--205. Springer, 2000.
\newblock \href {https://doi.org/10.1007/3-540-40064-8\_19}
  {\path{doi:10.1007/3-540-40064-8\_19}}.

\bibitem{Kotzig55}
Anton Kotzig.
\newblock Contribution to the theory of eulerian polyhedra.
\newblock {\em Math. Slovaca}, 5:101--113, 1955.

\bibitem{Lampis20}
Michael Lampis.
\newblock Finer tight bounds for coloring on clique-width.
\newblock {\em {SIAM} J. Discret. Math.}, 34(3):1538--1558, 2020.
\newblock \href {https://doi.org/10.1137/19M1280326}
  {\path{doi:10.1137/19M1280326}}.

\bibitem{Langer11}
Alexander Langer, Peter Rossmanith, and Somnath Sikdar.
\newblock Linear-time algorithms for graphs of bounded rankwidth: {A} fresh
  look using game theory - (extended abstract).
\newblock In Mitsunori Ogihara and Jun Tarui, editors, {\em Theory and
  Applications of Models of Computation - 8th Annual Conference, {TAMC} 2011,
  Tokyo, Japan, May 23-25, 2011. Proceedings}, volume 6648 of {\em Lecture
  Notes in Computer Science}, pages 505--516. Springer, 2011.
\newblock \href {https://doi.org/10.1007/978-3-642-20877-5\_49}
  {\path{doi:10.1007/978-3-642-20877-5\_49}}.

\bibitem{Libkin}
Leonid Libkin.
\newblock {\em Elements of Finite Model Theory}.
\newblock Texts in Theoretical Computer Science. An {EATCS} Series. Springer,
  2004.
\newblock URL: \url{http://www.cs.toronto.edu/\%7Elibkin/fmt}, \href
  {https://doi.org/10.1007/978-3-662-07003-1}
  {\path{doi:10.1007/978-3-662-07003-1}}.

\bibitem{NesetrilM06}
Jaroslav Nešetřil and Patrice~Ossona de~Mendez.
\newblock Linear time low tree-width partitions and algorithmic consequences.
\newblock In Jon~M. Kleinberg, editor, {\em Proceedings of the 38th Annual
  {ACM} Symposium on Theory of Computing, Seattle, WA, USA, May 21-23, 2006},
  pages 391--400. {ACM}, 2006.
\newblock \href {https://doi.org/10.1145/1132516.1132575}
  {\path{doi:10.1145/1132516.1132575}}.

\bibitem{Nesetril06}
Jaroslav Nešetřil and Patrice~Ossona de~Mendez.
\newblock Tree-depth, subgraph coloring and homomorphism bounds.
\newblock {\em Eur. J. Comb.}, 27(6):1022--1041, 2006.
\newblock \href {https://doi.org/10.1016/j.ejc.2005.01.010}
  {\path{doi:10.1016/j.ejc.2005.01.010}}.

\bibitem{sparsity}
Jaroslav Nešetřil and Patrice~Ossona de~Mendez.
\newblock {\em Sparsity - Graphs, Structures, and Algorithms}, volume~28 of
  {\em Algorithms and combinatorics}.
\newblock Springer, 2012.
\newblock \href {https://doi.org/10.1007/978-3-642-27875-4}
  {\path{doi:10.1007/978-3-642-27875-4}}.

\bibitem{Nesetril15}
Jaroslav Nešetřil and Patrice~Ossona de~Mendez.
\newblock On low tree-depth decompositions.
\newblock {\em Graphs Comb.}, 31(6):1941--1963, 2015.
\newblock \href {https://doi.org/10.1007/s00373-015-1569-7}
  {\path{doi:10.1007/s00373-015-1569-7}}.

\bibitem{Norine06}
Serguei Norine, Paul~D. Seymour, Robin Thomas, and Paul Wollan.
\newblock Proper minor-closed families are small.
\newblock {\em J. Comb. Theory, Ser. {B}}, 96(5):754--757, 2006.
\newblock \href {https://doi.org/10.1016/j.jctb.2006.01.006}
  {\path{doi:10.1016/j.jctb.2006.01.006}}.

\bibitem{Oum08}
Sang{-}il Oum.
\newblock Approximating rank-width and clique-width quickly.
\newblock {\em {ACM} Trans. Algorithms}, 5(1):10:1--10:20, 2008.
\newblock \href {https://doi.org/10.1145/1435375.1435385}
  {\path{doi:10.1145/1435375.1435385}}.

\bibitem{OumS06}
Sang{-}il Oum and Paul~D. Seymour.
\newblock Certifying large branch-width.
\newblock In {\em Proceedings of the Seventeenth Annual {ACM-SIAM} Symposium on
  Discrete Algorithms, {SODA} 2006, Miami, Florida, USA, January 22-26, 2006},
  pages 810--813. {ACM} Press, 2006.
\newblock URL: \url{http://dl.acm.org/citation.cfm?id=1109557.1109646}.

\bibitem{Pikhurko11}
Oleg Pikhurko and Oleg Verbitsky.
\newblock Logical complexity of graphs: a survey.
\newblock {\em Model theoretic methods in finite combinatorics}, 558:129--179,
  2011.

\bibitem{Plotkin94}
Serge~A. Plotkin, Satish Rao, and Warren~D. Smith.
\newblock Shallow excluded minors and improved graph decompositions.
\newblock In {\em Proceedings of the Fifth Annual {ACM-SIAM} Symposium on
  Discrete Algorithms. 23-25 January 1994, Arlington, Virginia, {USA}}, pages
  462--470, 1994.
\newblock URL: \url{http://dl.acm.org/citation.cfm?id=314464.314625}.

\bibitem{Seese96}
Detlef Seese.
\newblock Linear time computable problems and first-order descriptions.
\newblock {\em Mathematical Structures in Computer Science}, 6(6):505--526,
  1996.

\bibitem{Tamaki19}
Hisao Tamaki.
\newblock Positive-instance driven dynamic programming for treewidth.
\newblock {\em J. Comb. Optim.}, 37(4):1283--1311, 2019.
\newblock \href {https://doi.org/10.1007/s10878-018-0353-z}
  {\path{doi:10.1007/s10878-018-0353-z}}.

\end{thebibliography}

\end{document}